\PassOptionsToPackage{dvipsnames}{xcolor}
\documentclass[sigconf, nonacm]{acmart}
\usepackage[utf8]{inputenc}

\usepackage[dvipsnames]{xcolor}
\usepackage{booktabs}
\usepackage{enumitem}
\usepackage{graphicx}
\usepackage{threeparttable}
\usepackage{tikz}
\usepackage{xspace}
\usepackage[binary-units,group-separator={,}]{siunitx}
\usepackage{bbm}
\usepackage{hyperref}
\usepackage[capitalise]{cleveref}
\makeatletter
\def\plist@algorithm{Alg.\space}
\makeatother

\usepackage[ruled]{algorithm}
\usepackage[noend]{algpseudocode}

\setcopyright{none}
\newtheorem{theorem}{Theorem}

\newcounter{CommentCounter}

\newcommand{\clientsize}{{n_c}}

\newcommand{\serversizekth}{n_{s}^{i}}
\newcommand{\serversetnum}{N}

\newcommand{\tversky}{\text{Tv}}
\newcommand{\tverskyfull}{\tversky_{\alpha, \beta}}

\newcommand{\cpk}{pk}
\newcommand{\csk}{sk}
\newcommand{\collec}{\mathcal{Y}}
\newcommand{\msfunc}{\textsc{agg-process}}
\newcommand{\msclientfunc}{\textsc{reveal}}
\newcommand{\serverfunc}{\textsc{process}}
\newcommand{\clientfunc}{\textsc{reveal}}
\newcommand{\queryfunc}{\textsc{query}}
\newcommand{\psmserverfunc}{\textsc{match-process}}
\newcommand{\psmclientfunc}{\textsc{match-reveal}}

\newcommand{\querysdfunc}{\textsc{SD-}\queryfunc}
\newcommand{\querybothfunc}{\textsc{(SD-)}\queryfunc}
\newcommand{\malcheckfunc}{\textsc{mal-check}}
\newcommand{\malchecksdfunc}{\textsc{SD-}\malcheckfunc}
\newcommand{\malcheckbothfunc}{\textsc{(SD-)}\malcheckfunc}
\newcommand{\addvariant}[2]{{\textsc{#1-}#2}}
\newcommand{\psmvar}{\addvariant{match}{\serverfunc}}
\newcommand{\termisin}[1]{\encvar{s_{#1}}}
\newcommand{\termisinserver}[1]{\encvar{t_{#1}}}
\newcommand{\psmresp}{\encvar{\gamma_k}}
\newcommand{\auxdata}{\mathbb{A}}
\newcommand{\retauxdata}{\mathcal{D}}
\newcommand{\clientquery}{Q}
\newcommand{\cardinality}{\textsf{ca}}
\newcommand{\match}{Match\xspace}
\newcommand{\agg}{Agg\xspace}
\newcommand{\enccardinality}{\encvar{\cardinality}}
\newcommand{\tverskyParamTrans}{\textsf{Tversky-param-process}}
\newcommand{\compiuteCardinality}{$\enccardinality \leftarrow \sum_{i \in \NatNumUpTo{m}} \fheiszero(\termisin{i})$}

\newcommand{\tmin}{t_{\text{min}}}
\newcommand{\tmax}{t_{\text{max}}}

\newcommand{\msPSMans}[1]{\encvar{\gamma_{#1}}}
\newcommand{\msZeroAns}[1]{\encvar{b_{#1}}}
\newcommand{\msresponse}{\encvar{A}}
\newcommand{\retIndex}{\kappa}
\newcommand{\retCtr}[1]{\encvar{\text{ctr}_{#1}}}
\newcommand{\funcind}[1]{\mathbbm{1}[{#1}]}

\newcommand{\randin}{\mathrel{{\leftarrow}\vcenter{\hbox{\tiny\rmfamily\upshape\$}}}}

\newcommand{\secpar}{\ell}
 
\newcommand{\code}[1]{\text{\textbf{#1}}}
\newcommand{\codify}[1]{ \text{\textsf{#1}} }
\newcommand{\NatNumUpTo}[1]{[#1]}
\newcommand{\varAsList}[1]{{\langle #1 \rangle}}

\newcommand{\fheParamGen}{\textsf{HE.ParamGen}}
\newcommand{\fheKeyGen}{\textsf{HE.KeyGen}}
\newcommand{\fheEnc}{\textsf{HE.Enc}}
\newcommand{\fheDec}{\textsf{HE.Dec}}
\newcommand{\fheiszero}{\textsf{HE.IsZero}}
\newcommand{\fheisin}{\textsf{HE.IsIn}}

\newcommand{\fheeval}{\textsf{HE.Eval}}

\newcommand{\encvar}[1]{\llbracket{#1}\rrbracket}
\newcommand{\Zq}{\mathbb{Z}_q}
\newcommand{\param}{\textsf{params}}

\newcommand{\fheaddop}{+}
\newcommand{\fhesubop}{-}
\newcommand{\fhemultop}{\cdot}

\newcommand{\polyDegree}{N_{deg}}
\newcommand{\plainMod}{m_{pt}}
\newcommand{\ctxMod}{m_{ct}}
\newcommand{\bfvParam}[1]{P_{#1k}}
\newcommand{\clientCopyNum}{k}

\newcommand{\compequiv}{\stackrel{\text{c}}{\equiv}}
\newcommand{\simview}[1]{\text{View}_{\text{#1}}}
\newcommand{\sims}[1]{\mathcal{S}_{#1}}

\newcommand{\yes}{{\color{OliveGreen}\checkmark}}
\newcommand{\no}{{\color{red}$\times$}}

\newcommand{\order}{\mathcal{O}}

\newcommand{\tcclient}[1]{{\color{red}#1}}
\newcommand{\tcserver}[1]{{\color{OliveGreen}#1}}

\newcommand{\interalgspace}{\vspace{1.5mm}}

\newcommand{\parait}[1]{\noindent\textit{#1}} 

\newcommand{\para}[1]{\vspace{0.5mm}\textit{#1.}}
\newcommand{\parasec}[1]{\paragraph*{\textbf{#1}}}

\newcommand{\requirement}[2]{\hypertarget{req:#1}{}\para{#1: #2}}
\newcommand{\reqlink}[1]{\protect\hyperlink{req:#1}{#1}}
\theoremstyle{definition}
\newtheorem{definition}{Definition}
\newtheorem*{definition*}{Definition}

\newcommand{\tikzlongarrow}[2]{
\begin{tikzpicture}[]
  \draw[#1](0,0) -- node[above=-0.5ex]{\ensuremath{#2}} (1.5,0);
  \node[draw=none] (bottom) at (0,-0.5ex) {};
  \node[draw=none] (top) at (0,1ex) {};
  \node[draw=none] (lowerleft) at (bottom-|current bounding box.west) {};
  \node[draw=none] (topright) at (top-|current bounding box.east) {};
  \pgfresetboundingbox
  \draw[draw=none,use as bounding box] (lowerleft) rectangle (topright);
\end{tikzpicture}
}
\newcommand{\diagramsend}[1]{\tikzlongarrow{->}{#1}}
\newcommand{\diagramrecv}[1]{\tikzlongarrow{<-}{#1}}

\newcommand{\adv}{\mathcal{A}}

\begin{document}
  \author{Kasra EdalatNejad}
  \affiliation{\institution{EPFL}\country{Switzerland}}
  \email{kasra.edalat@epfl.ch}

  \author{Mathilde Raynal}
  \affiliation{\institution{EPFL}\country{Switzerland}}
  \email{mathilde.raynal@epfl.ch}

  \author{Wouter Lueks}
  \affiliation{\institution{CISPA Helmholtz Center for Information Security}\country{Germany}}
  \email{lueks@cispa.de}
  
  \author{Carmela Troncoso}
  \affiliation{\institution{EPFL}\country{Switzerland}}
  \email{carmela.troncoso@epfl.ch}

  \title{Private Collection Matching Protocols}

  \begin{abstract}
  {
    We introduce \emph{Private Collection Matching (PCM)} problems, in which a client aims to determine whether a collection of sets owned by a server matches their interests.
    Existing privacy-preserving cryptographic primitives cannot solve PCM problems efficiently without harming privacy.
    We propose a modular framework that enables designers to build privacy-preserving PCM systems that output one bit: whether a collection of server sets matches the client's set.
    The communication cost of our protocols scales linearly with the size of the client's set and is independent of the number of server elements.
    We demonstrate the potential of our framework by designing and implementing novel solutions for two real-world PCM problems: determining whether a dataset has chemical compounds of interest, and determining whether a document collection has relevant documents.
    Our evaluation shows that we offer a privacy gain with respect to existing works at a reasonable communication and computation cost.
  }
  \end{abstract}

  \maketitle

\section{Introduction}
\label{sec:intro}

In many scenarios, a server holds a collection of sets and clients wish to determine whether these server sets \emph{match} their own set, while both client and server keep their privacy. We call these \emph{Private Collection Matching (PCM)} problems. In this paper, we study for the first time the requirements of PCM problems. 

We identify the privacy and efficiency requirements of PCM problems by analyzing three real-world use cases: determining whether a pharmaceutical database contains compounds that are chemically similar to the client's~\cite{stumpfe2011, laufkotter2019, shimizu2015}, determining whether an investigative journalist holds relevant documents~\cite{DatashareNetwork} (or how many), and matching a user's profile to items or other users in mobile apps~\cite{thefork, stravaexploreroute, tinder}. We find that PCM problems have three common characteristics: 
(1) Clients want to compare their \emph{one} set with \emph{all} sets at the server. (2) Clients do not need per-server set results, only an aggregated output (e.g., whether any server set matches). (3) Clients and server want privacy: the server should learn nothing about the clients' set, and the clients should only learn the aggregated output. 
However, PCM problems differ in their definition of \emph{when sets match} and \emph{how to combine individual matching responses}.
Now, we discuss these two aspects in more details: \looseness=-1

\para{Set matching} 
Typically, set matching is defined as a function of the intersection of two sets. Hence, clients could detect a matching server set by using private set intersection (PSI) protocols~\cite{CristofaroGT12, CT09, CT10, Pinkas0TY19, PSZ14, PSZ18, PSWW18, Pinkas0TY19, KissnerS05} to privately compute the intersection, then post-process the intersection to determine interest locally. PCM applications differ in their matching criteria and may decide interest using measures such as a cardinality threshold, containment, or set similarity. This local processing approach, unfortunately, \emph{reduces privacy of the server's sets} by leaking information beyond the set's matching status to the client.
Such leakage could, for instance, reveal secret chemical properties of compounds, or the content of journalists' sensitive documents.

\para{Many-set} In PCM problems, the server holds \emph{a collection of $N$ sets}. This creates two challenges. First, running one matching (or PSI) interaction per server set is inefficient.
Second, revealing individual set-matching statuses harms server privacy.
While servers may be interested in selling data to or collaborating with clients, they want to ensure that clients cannot use `the PCM solution' to extract information about sets. Clients, meanwhile, often only need an aggregated response summarizing the utility of a collection.
Servers therefore enact application-depended aggregation policies ensuring that clients can, e.g., determine only whether at least one set matches or learn only the the number of matching sets.

We construct a framework that leverages computation \emph{in the encrypted domain} to solve PCM problems.
In our framework, shown in~\cref{fig:overview}, given an encrypted client set, the server uses a matching criteria $f_M$ to compute per-server-set \emph{binary} answers to ``is this set of interest to the client?''. 
Next, the server uses an aggregation policy $f_A$ to combine per server-set responses into a  collection-wide response. Finally, the client decrypts the aggregated result. \looseness=-1

\begin{figure}[tb]
    \centering
    \includegraphics[trim=6pt 2pt 7pt 0, width=0.9\linewidth, clip]{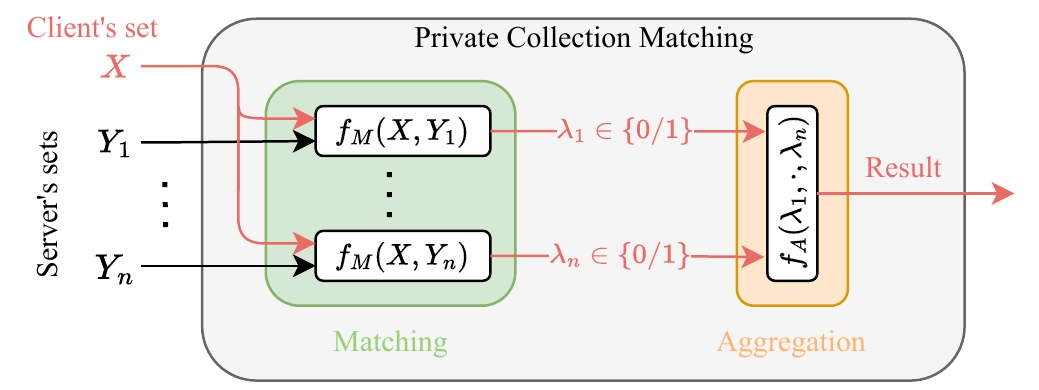}
    \caption{Structure of our PCM framework. Red arrows show values encrypted under the client's key. $f_M$ designates a matching function: it outputs a binary value $\lambda$ indicating whether two sets match. $f_A$ is an aggregation function that combines $n$ matching statuses  into a collection-wide result.}
    \label{fig:overview}
\end{figure}

Our work makes the following contributions:
\begin{itemize}[nosep, wide]
    \item[\checkmark] We introduce Private Collection Matching (PCM) problems. We derive their requirements from three real-world problems.

    \item[\checkmark] We design single-set protocols where the client learns a \emph{one-bit output} -- whether one server set is of interest to the client -- and many-set protocols where the client learns a \emph{collection-wide output} that aggregates individual matching responses. The communication cost of our protocols scales linearly with the size of the client's set and is \emph{independent} of the number of server sets and their total size. 

    \item[\checkmark] We propose a modular design that separates flexible set matching criteria from many-set aggregation. Our modularity enables extending our design with new matching or aggregation policies and simplifies building privacy-preserving PCM solutions. 

    \item[\checkmark] We demonstrate our framework's capability by solving chemical similarity and document search problems. 
    We show that our framework offers improved privacy with competitive cost compared to custom-made solutions, and significantly improving the latency, client's computation cost, and communication cost with respect to generic solutions that offer the same privacy guarantee.
\end{itemize}

\section{Private Collection Matching}
\label{sec:scenarios}
In this section, we define the Private Collection Matching (PCM) problem. We derive its basic requirements from three real-world matching problems. 
We also explain why existing PSI solutions cannot satisfy the privacy requirements of PCM problems.

\subsection{Case studies}
\label{sec:case-studies}
We study three cases that can benefit from PCM.

\para{Chemical similarity.}
Chemical research and development is a multi-billion dollar industry. When studying a new chemical compound, knowing the properties of similar compounds can speed up the research. In an effort to monetize research, companies sell datasets describing thousands to millions of compounds and their properties.
Chemical R\&D teams are willing to pay high prices for these datasets but \emph{only if} they include compounds similar to their research target. Determining whether this is the case is tricky: buyers want to hide the compound they are currently investigating~\cite{shimizu2015}, and sellers want to hide information about the compounds in their dataset before the sale is finalized.

Chemical similarity of compounds is determined by comparing molecular fingerprints of compounds~\cite{stumpfe2011, laufkotter2019, xue2001,willett1998,cereto2015,muegge2016}. Fingerprints are based on the substructure of compounds and are represented as fixed-size bit vectors -- these vectors are between a few hundred and few thousand bits long.
Measures such as Tversky~\cite{tversky1977features} and Jaccard~\cite{jaccard1912distribution} determine the similarity of these fingerprints, and thus of the compounds.

Revealing pair-wise intersection cardinalities or even similarity scores between the fingerprints of a target compound and the seller's compounds results in unacceptable leakage. A buyer can reconstruct an $n$-bit molecular fingerprint $F$ by learning similarity values between $F$ and $n+1$ known compounds~\cite{shimizu2015}.
To prevent inferences, the buyer should learn only the number of similar compounds in the seller's dataset, or, better, only learn whether at least one similar compound exists.

\para{Peer-to-Peer document search.}
Privacy-preserving peer-to-peer search engines help users and organizations with strict privacy requirements to collaborate safely. We take the example of investigative journalists who, while unwilling to make their investigations or documents publicly available, want to find collaboration opportunities within their network~\cite{DatashareNetwork}.

To identify those opportunities, a journalist performs a search to learn whether a document owner has documents of interest. 
A search query consists of keywords relevant to the journalist's investigation. The document owner compares the query to all documents in his collection. A document is deemed relevant if it contains \emph{all} or \emph{a sufficient number of} queried keywords.
Journalists own a collection of a thousand documents (on average), and each document is represented by around a hundred keywords.

The sensitivity of journalists' investigations demand that both the content of the documents and of the queries remain private~\cite{DatashareNetwork}. Journalists only need to learn one bit of information -- that at least one or a threshold number of documents in the owner's collection is relevant -- to determine whether they should contact the owner.

\para{Matching in mobile apps.}
A common feature in mobile apps is enabling users to find records of interest in the app servers' databases, e.g., restaurants~\cite{thefork}, routes for running~\cite{stravaexploreroute}, or suitable dating partners~\cite{tinder, OkCupid}. 
Users are typically interested in records that have at least a number of matching characteristics in common with their search criteria or that are a perfect match. Also, users need to be able to retrieve these records.

The user-provided criteria -- typically range choices entered via radio buttons or drop-down menus -- are compared to the attributes of records. An app database can have millions of records and records can have dozens to hundreds of attributes.

Both search criteria and records are sensitive. Knowing search criteria enables profiling of user interests. These are particularly sensitive for dating applications. Thus, search queries should be kept private. The secrecy of the records in the database is not only at the core of the business value of these apps but also required by law in cases where records contain personal data (e.g., dating apps).

\subsection{PCM requirements}
\label{sub:requirements}

We extract requirements that PCM protocols should fulfill based on the commonalities between the use cases. These come in addition to basic PSI properties such as client privacy.

\requirement{RQ.1}{Flexible set matching} \textit{PCM protocols need to be able to determine matches between sets without revealing other information such as intersections or cardinalities to the client.}
In the use cases, clients do not need to know the intersection or its cardinality. They are interested only in whether there is a match. Matches in these examples are a function of the intersection between a client and a server set: a chemical compound is a match when the Tversky or Jaccard similarity with the query exceeds a threshold; a document is a match when it contains some or all query keywords; and a record is a match when it includes a threshold of query attributes.
PCM protocols must be able to detect matching sets and compute a \emph{single} one-bit matching status per-set.

\requirement{RQ.2}{Aggregate many-set responses}
\textit{PCM protocols need to have the capability to provide an aggregated response for a collection of sets without leaking information about individual sets.}
Our use cases highlight 
that in many applications, a client (buyer, journalist, user) may want to compare their input (compounds under investigation, keywords of interest, search criteria) with a \emph{collection} of sets (compounds in a database, documents in a collection, records in a database). 
More importantly, we observe that clients wish to know how interesting the collection is as a whole.
For example, a buyer is interested in a chemical dataset if it contains at least one similar compound and a querier journalist may contact a document owner if the owner has a number of relevant documents in their collection.
Therefore, to satisfy clients' needs yet protect the server's privacy, PCM protocols should only reveal aggregated per-collection results.

\requirement{RQ.3}{Extreme imbalance}
\textit{PCM problems have thin clients and imbalanced input sizes; thus, protocols must not require communication and computation linear to the server's input size from clients.}
Drawing from our earlier scenarios, the total input size of the server may be as much as \emph{6 orders of magnitude} larger than the client's input, as shown by the chemical similarity scenario.
The server, holding many sets, can safely be assumed to be resourceful, while clients may be constrained in their capabilities, e.g., a client running the PCM protocol from their mobile phone. This can be in terms of computation, e.g., battery has to be preserved in mobile apps; or in terms of bandwidth, e.g., journalists that can be in locations with poor Internet access. Therefore, PCM protocols should not incur a large client-side cost.

\newcommand{\totalserversize}{N_s}

\subsection{Formal PCM definition}
\label{sub:formal-def}
Let $X$ be a client set with $\clientsize$ elements $\{x_1, \ldots, x_{\clientsize}\}$ from input domain $D$ and $\collec$ be a collection of $N$ server sets $\{Y_1, \ldots Y_N\}$ where the $i$'th server set $Y_i = \{y_{i, 1}, \ldots, y_{i, \serversizekth}\}$ has $\serversizekth$ elements, also from $D$, leading to a total server size of $\totalserversize = \sum_i \serversizekth$. 
We define two families of functions as follows: 

Matching functions $\lambda_i \leftarrow f_M (X, Y_i)$ take two sets $X$ and $Y_i$ as input and compute a binary matching status $\lambda_i$ determining interest. This family represents our flexible matching criteria \reqlink{RQ.1}. For example, in the document search scenario where a server set (document) is of interest when it contains all queried keywords, we define $f_M(X, Y_i)$ as 1 if $X \subseteq Y_i$ and 0 otherwise.

Aggregation functions $A \leftarrow f_A(\lambda_1, \ldots, \lambda_N)$ take $N$ binary matching statuses ($\lambda_i$) and aggregates them into a single response $A$. This family represents our aggregation requirement \reqlink{RQ.2}. For example, if we want to count the number of relevant documents in a search, we define $ f_A(\lambda_1, \ldots, \lambda_N)$ as $\sum_j \lambda_j$.

In \cref{tab:proto-summary} in \cref{ap:proto-summary}, we summarize the matching and aggregation functions that we implement.

\begin{definition}[PCM]
  \label{def:PCM}
  PCM protocols are two-party computations between a client and a server with common inputs $f_M$ and $f_A$, where the client learns an aggregated matching status and the server learns nothing. 
  Formally: 
  \[\left(A = f_A\left(f_M\left(X, Y_1\right), \ldots, f_M\left(X, Y_N\right)\right), \bot\right)  \leftarrow PCM_{f_M, f_A}(X, \collec)  \]
  We use this notation to define formal properties of PCM protocols.
\end{definition}

\begin{definition}[Correctness]
    \label{def:correctness}
    A PCM protocol is correct if the client output matches the result of $f_A(f_M(X, Y_1), \allowbreak \ldots, \allowbreak f_M(X, Y_N))$. 
\end{definition}

\begin{definition}[Client Privacy]
    \label{def:cl-priv}
    A PCM protocol is client private if the server cannot learn any information about the client's set beyond the maximum size of the client's set.
\end{definition} 

\begin{definition}[Server Privacy]
    \label{def:sv-priv}
    A PCM protocol is server private if the client cannot learn any information about the server elements beyond the number of server sets, the maximum server set size, and the explicit client output $A$.
\end{definition}

\section{Related Work}\label{sec:rel}
While ad-hoc solutions for chemical (Shimizu et al.~\cite{shimizu2015}) and document search (EdalatNejad et al.~\cite{DatashareNetwork}) exist, most prior work focuses on building private set intersection (PSI) protocols, which are a special case of PCM. 
We introduce the PSI protocols most relevant to our work. We leave the detailed comparison with ad-hoc solutions to the evaluation (see \cref{sec:eval}). 

We compare existing work on two critical aspects of PCM problems: privacy and efficiency. We summarize existing schemes and their suitability for the PCM scenario in \Cref{tab:related-work}.\looseness=-1

For privacy, we assess whether existing approaches provide flexible matching (\reqlink{RQ.1}) and aggregated many-set responses (\reqlink{RQ.2}). 
We note that the majority of prior works do not consider, or support, many sets. 
When there is no natural extension to support many sets at once, we run a single-set interaction per server set, leading to an  $N\times$ increase in cost. This naive extension does not provide the privacy enhancement of many-set aggregation but enables us to reason about the efficiency of these approaches.

For efficiency, taking into account the extreme imbalance requirement (\reqlink{RQ.3}), we focus on the client's computation and communication cost and require schemes to have a client cost of $\omega(\totalserversize (\approx N \serversizekth))$.

\begin{table}[tb]
  \caption{Overview of PSI approaches in the PCM setting.}
  \begin{tabular}{@{}l@{\hskip6pt}c@{\hskip6pt}c@{\hskip6pt}c@{}}
      \toprule
       & \multicolumn{2}{c}{Privacy} & Efficiency \\
      \cmidrule(lr){2-3} \cmidrule(lr){4-4}
       & \reqlink{RQ.1} & \reqlink{RQ.2} & \reqlink{RQ.3} \\
      \midrule
    OT \hfill \cite{FalkNO19, PSZ14, PSSZ15, PSZ18, KKRT16, PinkasRTY19}             & \no  & \no   & \no  \\ 
    Asym. OPRF \hfill\cite{CT09, CT10, KLS17, CristofaroGT12, RosulekT21} & \no  & \no   & \yes$^*$ \\
    OPE \hfill \cite{FNP04, H15, DRMY12, CLR17} & \no  & \no   & \yes \\
    Generic SMC \hfill \cite{HEK12, KRS19,  PinkasSSW09}        & \yes & \yes  & \no  \\
    Circuit-PSI \hfill \cite{KarakocK20, CiampiO18, PSWW18, Pinkas0TY19, ChandranGS22, RindalS21,MaC22}  & \yes & \yes   & \no  \\
    Flexible func.\hfill\cite{YingCPXL22, 0001C18, GhoshS19, IonKNPSSSY17, ZhaoC17, shimizu2015} & \yes & \no   & \no  \\
      \textbf{This paper} & \yes    & \yes  & \yes \\
      \bottomrule
  \end{tabular}
  \caption*{\footnotesize ${ ^*}$Efficient communication requires pre-processing.}
  \label{tab:related-work}
\end{table}

\parasec{Traditional single-set PSI} \label{sub:rel-work-psi}
We first study protocols \emph{solely focusing on single-set intersection or cardinality}.
We study schemes that offer \emph{enhanced functionality or privacy} separately.

In PSI, clients learn information about the intersection of two sets while (i) not learning anything about the server's non-intersecting elements, and (ii) not leaking any information about their own set to the server. PSI protocols in the literature focus on providing two possible outputs: the \emph{intersection} (e.g., finding common network intrusions~\cite{NMH10}, or discovering contacts~\cite{demmler2018pir}); and the \emph{cardinality of the intersection} (e.g., privately counting the number of common friends between two social media users~\cite{NCD13}, or performing genomic tests~\cite{baldi2011countering}).
Works in this area opt for a variety of trade-offs between the computational capability and the amount of bandwidth required to run the protocol~\cite{CLR17, KLS17, PinkasRTY20, KRS19}.
These works show that PSI can scale to large datasets~\cite{Pinkas0TY19, kamara2013scaling}, and support light clients~\cite{KLS17}.

We classify PSI protocols according to the cryptographic primitive they use to implement the intersection: 

\para{OT-based protocols}
The fastest class of PSI protocols builds on oblivious transfers (OTs)~\cite{FalkNO19, PSSZ15, PSZ18, KKRT16, PSZ14,PinkasRTY19}. 
These protocols use hashing to map elements to small buckets and then apply efficient oblivious PRFs based on OTs to the items in these buckets to enable one-to-one plaintext comparisons.
Typical OT-based approaches reveal comparison results (and thus intersections or cardinalities), to the client; therefore, they do not satisfy our single-set privacy requirement (\reqlink{RQ.1}). As these approaches reveal individual detailed set responses, \emph{private} aggregation (\reqlink{RQ.2}) is impossible.
The communication cost is linear in the size of the client \emph{and} server set, and consequently, linear in the total server size $\totalserversize$; thus, they do not satisfy our efficiency requirement (\reqlink{RQ.3}).
In \cref{ap:spot-bench}, we confirm our efficiency assessment by evaluating the cost of SpOT-light~\cite{PinkasRTY19}, one of the fastest OT-based PSI protocols, in the PCM setting. We show significant improvement in latency (\numrange{10}{65}x), client's computation (\numrange{1800}{24800}x), and transfer cost (\numrange{1.7}{27}x).

\para{OPRFs using asymmetric cryptography}
Some protocols construct Oblivious Pseudo-random Functions (OPRFs) from asymmetric primitives such as Diffie-Hellman~\cite{RosulekT21, KLS17},  RSA~\cite{CT09, CT10, KLS17} or discrete logarithms~\cite{CristofaroGT12}. The client obliviously evaluates the PRF on its elements with the server, and the server sends the PRF evaluation of its elements to the client. The client then locally compares.

OPRF-based approaches rely on comparison, similar to OT, and cannot compute flexible set matches without leaking intermediate data nor do they support aggregation. The communication cost is linear in the size of client \emph{and} server sets. 
But, preprocessing can make the transfer cost independent of the server size~\cite{KLS17}.

\para{Oblivious polynomial evaluation}
Another approach is to use (partial) homomorphic encryption~\cite{RAD78} to determine set intersection using oblivious polynomial evaluation (OPE)~\cite{FNP04, H15, DRMY12, CLR17}. The client encrypts their elements and sends them to the server. The server constructs a polynomial with its set elements as roots, evaluates the polynomial on encrypted client elements, randomizes the result, and sends them back to the client. The client decrypts the results, a 0 indicates a matching element. We use a similar approach in our schemes. Existing OPE-based approaches do not support flexible set matching or aggregation, but achieve cost independent of the server input size ($\totalserversize$) for the client.

\para{Generic SMC}
PSI protocols based on SMC tools~\cite{BeaverMR90,Yao86} construct full circuits such as sort-compare-shuffle~\cite{HEK12} to compute the intersection~\cite{KRS19,  PinkasSSW09}. They can be extended to support flexible set matching or many-set aggregation. However, circuits have communication linear in the size of their inputs (wires) which guarantees a transfer cost of $\order(\clientsize + \totalserversize)$.  This is a fundamental limit. Thus, circuits cannot satisfy our efficiency requirement (\reqlink{RQ.3}).
Because circuits can satisfy our privacy requirements, we develop a generic alternative to our framework using an SMC compiler in \cref{sub:doc-search-eval} and show that our system improves latency (\numrange{2}{96}x), client's computation (\numrange{75}{2250}x), and transfer cost (\numrange{93}{2800}x).
Besides for generic protocols, circuits are used to (1) extend OT-based protocols (discussed as `Circuit-PSI' below) or (2) obliviously evaluate PRFs~\cite{KRS19, KLS17} such as AES or LowMC~\cite{AlbrechtR0TZ15} (discussed as `OPRF').

\parasec{Custom PSI protocols}\label{sub:rel-beyond-intersection} Some PSI protocols go beyond cardinality and compute more complex functions over the intersection.

\para{Circuit-PSI}
A new line of work extends hashing- and OT-based PSI protocols to support arbitrary extensions with generic SMCs~\cite{CiampiO18, PSWW18, Pinkas0TY19, ChandranGS22, RindalS21, KarakocK20, MaC22}.
These works compute the intersection of two sets but instead of revealing the plain result to the client, they secret share the intersection between the two parties. This secret shared output enables parties to privately compute arbitrary functions on top of the intersection.
Unfortunately, these work focus on scenarios with equal client and server sizes as their cost is linear in the input of both parties $\order(\clientsize + \totalserversize)$. This linear communication cost is a fundamental limit.
As these approaches can satisfy our privacy requirement, we evaluate Chandran et al.~\cite{ChandranGS22}, a state-of-the-art Circuit-PSI paper, in the PCM setting in \cref{sub:doc-search-eval}. We show that our framework significantly improves latency (\num{580}x), client's computation (\num{70000}x), and transfer cost (\num{2360}x).

\para{Flexible functionality}
Several privacy-preserving custom protocols provide functionality beyond computing intersection or cardinality. For instance, computing the sum or statistical functions over associated data~\cite{IonKNPSSSY17, YingCPXL22}, evaluating a threshold on intersection size~\cite{0001C18, GhoshS19, ZhaoC17}, or computing Tversky similarity~\cite{shimizu2015}. These approaches improve privacy by supporting flexible set matching but do not extend well to many-sets scenarios and aggregation.
They are optimized for a specific setting and do not achieve cost independent of the server input size.

\parasec{Orthogonal works}\label{sub:rel-other-works}
We briefly mention two groups of related work that, while of interest, are orthogonal to PCM problems and solve different challenges. 
\emph{Encrypted databases and ORAM}~\cite{PoddarBP19, PopaRZB11, StefanovDSCFRYD18} let a client query their own outsourced data. This cannot address PCM problems, where the server holds private information in addition to the client, and wishes to minimize what the clients learn. 
In \emph{Multi-party PSI}~\cite{BayEAV21, WangBU21, KolesnikovMPRT17} $p$ parties each with a set $S_i$ compute one intersection $I = \bigcap S_i$. PCM problems, instead, are a two-party protocol where the server holds $N$ sets.

\section{A Framework for PCM Schemes}
\label{sec:framework}

\begin{figure}[tb]
  \centering
  \includegraphics[width=\linewidth, clip]{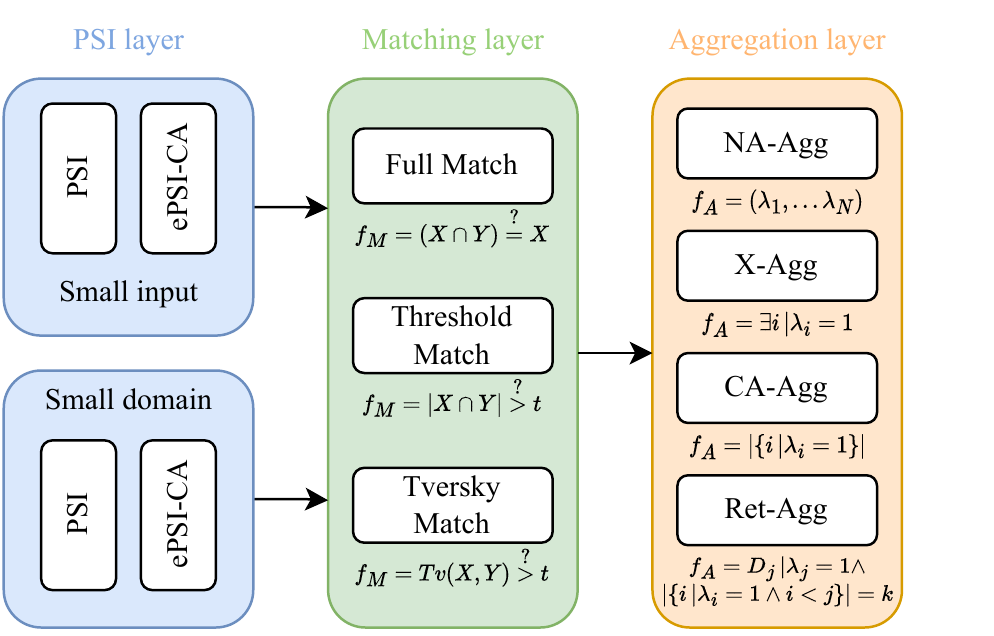}
  \caption{An overview of our layers and their composition. Refer to \cref{tab:proto-summary} for a summary of protocol definitions. 
  }
  \label{fig:psi_layers}
\end{figure}

We showed that existing work -- except for ad-hoc solutions -- cannot solve PCM problems without losing either privacy or efficiency. We now introduce a modular framework that enables the design of PCM solutions with minimal effort and strong privacy while providing
performance close to ad-hoc solutions.

The framework has three layers, shown in \cref{fig:psi_layers}:

\begin{description}[font=\normalfont\itshape]
\item [PSI layer] protocols operate directly on a client's and a server's set. These protocols compute single-set PSI functionalities such as intersection or cardinality. Our implementation focuses on two scenarios: small input domain size and small constant-size client sets.

\item [Matching layer] protocols use PSI layer protocols to compute a binary answer determining whether each of the server sets matches the interest of the client according to a pre-defined matching function $f_M$ (\reqlink{RQ.1}). Computation in this layer is the same regardless of the scenario chosen in the PSI layer.
\item [Aggregation layer] protocols aggregate $N$ single-set responses into one collection-wide response according to a pre-defined aggregation function $f_A$ (\reqlink{RQ.2}). This layer achieves constant size responses and ensures client's communication efficiency (\reqlink{RQ.3}).
\end{description}

\para{Modularity}
While we provide a large number of protocols for the PSI, Matching, and Aggregations layers (see \cref{fig:psi_layers}), an advantage of our framework is extensibility.
Whenever an application requires new matching or aggregation criteria, designers can add (or adapt) a single functionality while taking advantage of the existing optimized layers.
As an example, we extend our framework to support the single-set PSI-SUM functionality in \cref{ap:psi-sum}.

The layers can also be used as standalone protocols. 
We include blocks such as `Naive' aggregation such that even if an application does not require matching or aggregation the designer can use the framework to enjoy its capability to tackle the many-set scenario.

\para{Security goals} 
Our framework must enforce the privacy of clients and servers against malicious adversaries. We have a more relaxed goal for correctness. Servers are free to choose their input, allowing them to degrade the quality of the protocol's result without any misbehavior. We accept this inherent weakness of PCM protocols and make the deliberate decision to aim for correctness in the semi-honest model only, and \emph{not} when the server is malicious.

\section{Technical Background}
We introduce our notation and define the syntax of the fully homomorphic encryption scheme we use.

\para{Notation} Let $\secpar$ be a security parameter. We write $x \randin X$ to denote that $x$ is drawn uniformly at random from the set $X$. Let $q$ be a positive integer, then $\Zq$ denotes the set of integers $[0, \ldots, q)$, and $\Zq^*$ represent the elements of $\Zq$ that are co-prime with $q$.
We write $\NatNumUpTo{n}$ to denote the set $\{1, \ldots, n\}$, and use
$\varAsList{a_i}_m$ to present the list $[a_1,  \ldots, a_m]$. We drop the
subscript $m$ when the list length is clear from the context. We write $\encvar{x}$ to denote the encryption of  $x$.
We write $\mathbbm{1}[E]$ to refer to the indicator function that returns `1' if the event $E$ is true, and `0' otherwise.
\Cref{tab:notation}  in \cref{ap:notation} summarizes our notation.

\subsection{Homomorphic Encryption}
\label{sub:fhe}
Homomorphic encryption (HE) schemes enable arithmetic operations on encrypted values without decryption. We use HE schemes that operate over the ring $\Zq$ with prime $q$, such as BFV~\cite{FanV12}.

\para{Syntax} HE is defined by the following procedures:

\begin{itemize}[nosep]
\item $\param \leftarrow \fheParamGen(q)$. Generates HE parameters with the plaintext domain $\Zq$.
\item $pk, sk \leftarrow \fheKeyGen(\param)$. Takes the parameters $\param$ and generates a fresh pair of keys $(pk, sk)$. For brevity, we do not explicitly mention evaluation keys $evk$ and consider them to be incorporated in the public key.
\item $\encvar{x} \leftarrow \fheEnc(pk, x)$. Takes the public key $pk$ and a message $x \in \Zq$ and returns the ciphertext $\encvar{x}$.
\item $x \leftarrow \fheDec(sk, \encvar{x})$. Takes the secret key $sk$ and a ciphertext $\encvar{x}$ and returns the decrypted message $x$.
\end{itemize}
The correctness property of homomorphic encryption ensures that $\fheDec(sk, \fheEnc(pk, x)) \equiv x \pmod q$.

\para{Homomorphic operations} HE schemes support homomorphic addition (denoted by~$\fheaddop$) and subtraction (denoted by~$\fhesubop$) of ciphertexts: $\fheDec(\encvar{a} \fheaddop \encvar{b}) = a + b \bmod q$ and $\fheDec(\encvar{a} \fhesubop \encvar{b}) = a - b \bmod q$.
HE schemes also support multiplication (denoted by~$\fhemultop$) of ciphertexts: $\fheDec(\encvar{a} \fhemultop \encvar{b}) = a \cdot b \bmod q$.

Besides operating on two ciphertexts, it is possible to perform addition and multiplication with plaintext scalars. In many schemes, such scalar-ciphertext operations are more efficient than first encrypting the scalar and then performing a standard ciphertext-ciphertext operation.
We abuse notation and write $a\encvar{x} + b$ to represent $(\encvar{a}\cdot\encvar{x}) + \encvar{b} =\encvar{ax + b}$.

\para{Multiplicative depth} 
Our framework is designed with fully homomorphic encryption (FHE) in mind and assumes unbounded multiplication depth. For practical purposes, we port the majority, but not all, of our protocols to support execution with somewhat homomorphic encryption (SWHE) and optimize operations in \cref{sec:practical}.\looseness=-1

\subsection{Core functions}
The complex functionality of PCM protocols can be reduced to a sequence of zero detection and inclusion test procedures.
These two functions allow us to describe our protocol at a higher abstraction level. Moreover, any improvement to these basic functions automatically enhances our framework.

\para{Zero detection} The function $\encvar{b} \leftarrow \fheiszero(pk, \encvar{x})$ computes whether the ciphertext $\encvar{x}$ is an encryption of zero. The binary output $b \in \{0, 1\}$ is defined as $b = 1$ if $x \equiv 0 \pmod q$ otherwise $b = 0$.
\begin{algorithm}[tbp]
    \footnotesize
    \caption{Check whether $\encvar{x}$ is zero.}
    \label{alg:bool}
    \begin{algorithmic}
        \Function{$\fheiszero$}{$pk, \encvar{x}$}
        \State $\encvar{b} = 1 - \encvar{x}^{(q-1)}$
        \Comment $b \leftarrow (x = 0)$
        \State \Return  $\encvar{b}$
    \EndFunction
    \end{algorithmic}
\end{algorithm}
We use Fermat's Little Theorem for zero detection~\cite{BurkhartSMD10}. 
We rely on the prime ring structure of $\Zq$  as any non-zero variable $x \in \Zq^*$ to the power $q - 1$ is congruent to one modulo the prime $q$. We can perform this exponentiation with $\lg(q)$ multiplications.
See Algorithm~\ref{alg:bool} for the implementation.

The high multiplicative depth of $\fheiszero$ makes it impractical for use with most SWHE schemes. We hope that the research and advances in HE comparison enables efficient instantiations of this function and unlock our framework's full capabilities. When evaluating our framework in \cref{sec:eval}, we use ad-hoc techniques to prevent the need for this function.

\para{Inclusion test} The function $\encvar{I} \leftarrow \fheisin(pk, \encvar{x}, Y)$ checks if $x$ is included in the set $Y$ of cardinality $n$.
We consider two variants. In the first, $Y$ is a set of ciphertexts $\encvar{y_i}$, in the second, $Y$ is a set of plaintexts $y_i$. In both cases, the output $I$ equals 0 if and only if a $y_i$ exists such that $x \equiv y_i \pmod q$, otherwise $I$ will be a uniformly random element in $\Zq^*$.
See Algorithm~\ref{alg:isin} for the implementation.

\newcommand{\isinprodval}{I}
\begin{algorithm}[tbp]
    \footnotesize
    \caption{Check inclusion of an encrypted variable $x$ in a plain $Y = \{y_1, \ldots, y_n\}$ or an encrypted $Y=\{\encvar{y_1}, .., \encvar{y_n}\}$ set.}
    \label{alg:isin}
    \begin{algorithmic}
        \Function{$\fheisin$}{$pk, \encvar{x}, Y=\{\encvar{y_1}, .., \encvar{y_n}\}$}
            \State $\encvar{\isinprodval} \leftarrow \prod_{i \in \NatNumUpTo{n}}(\encvar{x} \fhesubop \encvar{y_i})$
            \State $r \randin \Zq^*$
            \State \Return  $r \fhemultop \encvar{\isinprodval}$
        \EndFunction
        \Function{$\fheisin$}{$pk, \encvar{x}, Y=\{y_1, .., y_n\}$}
            \State $[a_0, \ldots, a_n] \leftarrow \code{ToCoeffs}(Y)$ \Comment{Such that $\prod_i(\mathrm{x} \fhesubop y_i) = \sum_i a_i \mathrm{x}^i $}
            \State $\encvar{\isinprodval} \leftarrow \sum_{i \in [0 \ldots n]} a_i \cdot \encvar{x}^i $
            \State $r \randin \Zq^*$
            \State \Return  $r \fhemultop \encvar{\isinprodval}$
        \EndFunction
    \end{algorithmic}
\end{algorithm}

The function $\fheisin$ relies on oblivious polynomial evaluation (OPE)~\cite{FNP04,H15}. We create an (implicit) polynomial $P$ with roots $y_i$, and evaluate $\encvar{\isinprodval} = \encvar{r\cdot P(x)}$. If $x$ is in the set, there exists a variable $y_i$ where $x \equiv y_i$, thus $\isinprodval$ is zero. Otherwise, $\isinprodval$ is the product of $n$ non-zero factors modulo $q$. Since $q$ is prime, the product of non-zero values is non-zero. The random value $r$ ensures uniformity in this case. 
The multiplicative depth of $\fheisin$ scales with the size of $Y$. We use the second form, where $Y$ is a set of plaintexts, to lower the multiplicative depth when $\encvar{x^i}$ are known, see \cref{sub:fhe-optimization}.

\section{PSI Layer}
\label{sec:base-layer}
The PSI layer of our framework implements basic PSI functionalities: computing intersection or intersection cardinality.
These protocols can be used in isolation, but in our framework they serve to form the input to the matching layer (see \cref{fig:psi_layers}).
We build PSI protocols for two scenarios: (1) scenarios where the client set has a small constant size (e.g., document search queries which typically have less than 10 keywords); and (2)  scenarios  where set elements come from a small input domain (e.g., gender and age in a dating profile).
First, we assume semi-honest clients and build basic protocols. In \cref{sub:base-malicious} we secure our protocols against malicious clients.

We structure our single-set protocols following \cref{fig:single_psi_proto}.
The client generates a HE key pair $(\cpk, \csk) \leftarrow \fheKeyGen(\param)$ and sends the public key $\cpk$ to the server ahead of the protocol. Clients perform $\queryfunc$ and send an encrypted representation of their set to the server. The server runs a protocol-specific processing function $\serverfunc$ to obtain the result $M$. The protocol is either used as the first layer and passes $M$ into the second layer, or is stand-alone and returns $M$ to the client. The server optionally runs $\malcheckfunc$ to randomize the result of malicious queriers ($\encvar{R}$ is zero for honest queries). Finally, the client runs the protocol-specific function $\clientfunc$ to compute the result. 
We denote algorithms run by the \tcclient{client} in \tcclient{red} and by the \tcserver{server} in \tcserver{green}.
We use \textcolor{blue}{blue} to show the \emph{optional} \textcolor{blue}{server-side} protection against malicious clients.

\subsection{Small constant-size client set}
\label{sub:base-proto}
We start with scenarios where client sets are small and constant-size, typical for representing a search criteria. 
Clients use $\queryfunc$ to encrypt their set elements $x_i \in X$ as a query $Q$ and send it to the server.
Algorithm~\ref{alg:single_psi_algs} instantiates small input functions.

\para{PSI} The PSI protocol computes $\text{PSI}(X, Y_k) = X \cap Y_k$. The server uses the inclusion test $\fheisin$ (see \cref{sub:fhe}) to compute an inclusion status $\termisin{i}$ for each client element $x_i$ (see $\addvariant{PSI}{\serverfunc}$). An element $x_i$ is in the intersection if and only if the corresponding inclusion status $\termisin{i}$ is zero (recall that the inclusion test produces zero for values \emph{in} the set).
When used as a stand-alone protocol, the server returns the list of encrypted inclusion values $M$, which the client then decrypts (see $\addvariant{PSI}{\clientfunc}$).

\para{Cardinality}
The PSI cardinality protocols compute $\text{PSI-CA}(X, Y_k) \allowbreak= |X \cap Y_k|$. There exist two variants: the standard PSI-CA variant in which the \emph{client} learns the cardinality $|X \cap Y_k|$~\cite{CristofaroGT12, DebnathD15}, and the ePSI-CA variant in which the \emph{server} learns an \emph{encrypted} cardinality~\cite{0001C18}.
We focus on the latter to enable further computation on the intersection cardinality in the next layers.

Our ePSI-CA protocol (see $\addvariant{ePSI-CA}{\serverfunc}$) first computes the inclusion statuses $\termisin{i}$ using $\addvariant{PSI}{\serverfunc}$ and then uses
$\fheiszero$ to compute -- in the ciphertext domain -- the cardinality, i.e., the number of elements $\termisin{i}$ that are zero.
When used as a stand-alone protocol, the server returns $\enccardinality$ to the client which decrypts it to obtain the answer (see $\addvariant{ePSI-CA}{\clientfunc}$).

When the cardinality protocol is used as a stand-alone protocol without next layers, it is possible to mimic earlier work~\cite{CristofaroGT12} and construct a cardinality protocol from the above-mentioned naive PSI protocol by shuffling server responses $M$ before returning them.

\para{Efficiency}
 While literature often dismisses OPE-based schemes due to their `quadratic' total computation cost $\order(|X|\cdot|Y_k|)$, this approach excels in PCM scenarios with small client input.
Our protocols achieve client computation and communication costs of $\order(\clientsize)$, which is independent of the server's input size. While the `extra' burden for the server is linear in the size of the client set which is a small constant.

\begin{figure}[tb]
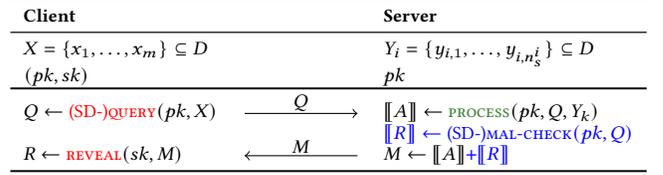

    \centering
    \footnotesize
    \begin{tabular}{l@{\hskip8pt}c@{\hskip8pt}l}
        \toprule
        \textbf{Client} & & \textbf{Server} \\
        \midrule
        $ X = \{x_1, \ldots, x_m\} \subseteq D $ & & $ Y_i = \{y_{i, 1}, \ldots, y_{i, \serversizekth}\} \subseteq D $ \\
        $ (\cpk, \csk) $ & & $\cpk$ \\
        \specialrule{\heavyrulewidth}{0pt}{5pt}
        $Q \gets \tcclient{\textsc{(SD-)}\queryfunc}(\cpk, X)$ &
            \diagramsend{Q }
           & $\encvar{A} \leftarrow \tcserver{\serverfunc}(\cpk, Q, Y_k)$ \\
            & & \textcolor{blue}{$\encvar{R} \leftarrow \textsc{(SD-)}\malcheckfunc(\cpk, Q)$} \\
        $R \leftarrow \tcclient{\clientfunc}(\csk, M)$ & \diagramrecv{M}  &  $M \leftarrow \encvar{A} \textcolor{blue}{+ \encvar{R}}$  \\
        \bottomrule
    \end{tabular}
    \vspace{-2ex}
    \caption{Single-set protocol structure. SD refers to small domain. \textcolor{blue}{Blue parts} (optional) protect against malicious clients.}
    \label{fig:single_psi_proto}
\end{figure}

\begin{algorithm}[tb]
    \footnotesize
    \caption{Single set procedures with \emph{small input size}.}
    \label{alg:single_psi_algs}
    \begin{algorithmic}
        \Function{\tcclient{$\queryfunc$}}{$\cpk, X$}
            \State $\encvar{{x_i}} \leftarrow \fheEnc(\cpk, x_i)$
            \State \Return  $Q = \varAsList{\encvar{{x_i}}}$
        \EndFunction
        \interalgspace

        \Function{\tcserver{$\addvariant{PSI}{\serverfunc}$}}{$\cpk, Q = \varAsList{\encvar{{x_i}}} , Y_k$}
            \State $\termisin{i} \leftarrow \fheisin(\cpk, \encvar{{x_i}}, Y_k)$
            \State \Return $M \leftarrow \varAsList{\termisin{i}}$
        \EndFunction
        \Function{\tcclient{$\addvariant{PSI}{\clientfunc}$}}{$\cpk,  M = \langle \termisin{i} \rangle$}
            \State \Return  $\{x_i \mid \fheDec(\csk, \termisin{i}) = 0\}$
        \EndFunction
        \interalgspace

        \Function{\tcserver{$\addvariant{ePSI-CA}{\serverfunc}$}}{$\cpk,  Q = \varAsList{\encvar{{x_i}}}, Y_k$}
            \State $\varAsList{\termisin{i}} \leftarrow \addvariant{PSI}{\serverfunc}(\cpk,  \langle \encvar{{x_i}} \rangle, Y_k)$
            \State \compiuteCardinality
            \State \Return $M \leftarrow \enccardinality $
        \EndFunction
        \Function{\tcclient{$\addvariant{ePSI-CA}{\clientfunc}$}}{$\cpk,  M = \enccardinality$}
            \State \Return  $\fheDec(\csk, \enccardinality)$
        \EndFunction
        \interalgspace

        \Function{\textcolor{blue}{$\malcheckfunc$}}{$\cpk, Q = \varAsList{\encvar{{x_i}}}$}
            \State  $\encvar{T} \leftarrow \prod_{i \in \NatNumUpTo{|Q|}, j \in \NatNumUpTo{i-1}} (\encvar{x_i}-\encvar{x_j})$
            \State $\encvar{R} \leftarrow r \cdot \fheiszero(\cpk, \encvar{T})$
            \State \Return $\encvar{R}$
        \EndFunction

    \end{algorithmic}
\end{algorithm}

\subsection{Small input domain}
\label{sub:small-domain}
When the set's input domain $D$ is small, sets can be efficiently represented and manipulated as bit-vectors~\cite{BayEAV21, RuanWMZ19, shimizu2015}. Parties agree on a fixed ordering $d_1, \ldots, d_{|D|}$ of elements of $D$.  
Then, clients use $\querysdfunc$ to compute a vector of encrypted inclusion statuses $\varAsList{\encvar{z_i}}$, where $z_i = 1$ iff $d_i \in X$ and 0 otherwise, for all $d_i$s in $D$.
We instantiate small domain procedures, except for $\clientfunc$ processes that are not impacted by the domain size, in Algorithm~\ref{alg:single_psi_algs_sd}.

\begin{algorithm}[tb]
    \footnotesize
    \caption{Single set procedures with \emph{small input domain}.}
    \label{alg:single_psi_algs_sd}
    \begin{algorithmic}
        \Function{\tcclient{$\querysdfunc$}}{$\cpk, X$}
            \State $z_i \leftarrow  (d_i \in X)$ 
            \State $\encvar{{z_i}} \leftarrow \fheEnc(\cpk, z_i)$ 
            \State \Return  $Q = \varAsList{\encvar{{z_i}}}$
        \EndFunction
        \interalgspace

        \Function{$\tcserver{\addvariant{PSI-SD}{\serverfunc}}$}{$\cpk, Q = \varAsList{\encvar{z_i}}, Y_k$}
            \State $v_i \leftarrow (d_i \in Y_k)$
            \State $\termisin{i} \leftarrow \encvar{z_i} \cdot v_i$
            \State \Return $M \leftarrow \varAsList{\termisin{i}}$
        \EndFunction
        \interalgspace
        
        \Function{$\tcserver{\addvariant{ePSI-CA-SD}{\serverfunc}}$}{$\cpk,   Q = \varAsList{\encvar{z_i}}, Y_k$}
            \State $\varAsList{\termisin{i}} \leftarrow \addvariant{PSI-SD}{\serverfunc}(\cpk,  \langle \encvar{{z_i}} \rangle, Y_k)$
            \State $\enccardinality \leftarrow  \sum_{d_i \in D} \termisin{i}$
            \State \Return $M \leftarrow \enccardinality $
        \EndFunction
        \interalgspace

        \Function{\textcolor{blue}{$\malchecksdfunc$}}{$\cpk, Q = \varAsList{\encvar{{z_i}}}$}
            \State  $\encvar{t_i} \leftarrow \fheisin(\cpk, \encvar{z_d}, \{0, 1\})$
            \State $\encvar{R} \leftarrow \sum_{d \in D} \encvar{t_i}$
            \State \Return $\encvar{R}$
        \EndFunction
    \end{algorithmic}
\end{algorithm}

The function $\addvariant{PSI-SD}{\serverfunc}$ creates a bit vector of the intersection. 
The status $\termisin{i}$ is an encryption of 1 if $d_i$ is present in both sets and 0 otherwise.
To do so, the server multiplies the indicator $\encvar{z_i}$ with another binary indicator $v_i$ determining whether the element $d_i$ is present in the server set $Y_k$. 
The function $\addvariant{ePSI-CA-SD}{\serverfunc}$ computes the sum of the inclusion statuses $\termisin{i}$ for all domain values $d_i \in D$. 
The functions $\queryfunc$, $\malcheckfunc$, and base layer $\serverfunc{}$ must have the same domain size. The rest of the functions and layers are not impacted by the choice of domain.

\para{Efficiency}
While the idea of representing sets as bit-vectors is not new~\cite{BayEAV21, RuanWMZ19, shimizu2015}, existing works dismiss FHE protocols as too costly and focus on additively homomorphic solutions.
We use the inherent parallelism of schemes such as BFV~\cite{FanV12}, that we discuss in \cref{sec:practical}, to achieve lower computation and communication costs, especially in the many-set scenario. In \cref{ap:sd-bench}, we show that PSI-CA-SD has a competitive performance to the existing schemes such as Ruan et al.~\cite{RuanWMZ19} and provides better privacy.

\subsection{Malicious queries}\label{sub:base-malicious}
The previous sections assumed semi-honest clients that perform $\textsc{(SD-)}\queryfunc$ correctly. Now, we allow misbehaving clients and provide protocols that protect servers against malicious queries. More specifically, we ensure that each query maps to a \emph{valid set} $X$, under the assumption that each ciphertext either decrypts to a unique scalar plaintext value or the decryption fails.\footnote{We discuss issues arising from using HE in practice, and how they limit the security of our implementation in \cref{ap:circuit-priv}.}

Malicious clients can deviate from the protocol to learn more than permitted.
Since the query is encrypted, the server cannot detect this misbehavior. 
Zero-knowledge proofs are not practical in the FHE setting, so we rely on an HE technique to ensure honest behavior from malicious clients. 
The server uses $\malcheckbothfunc$ to compute a randomizer term $\encvar{R}$ that is random in $\Zq$ if the client misbehaves and 0 otherwise. 
By adding $\encvar{R}$ to the result $\encvar{A}$, misbehaving clients learn nothing about the real result. 
With abuse of notation, the server adds a vector of fresh randomizers when the result is a vector such as the output of \addvariant{PSI-SD}{\serverfunc}.
The server can amortize the cost of computing $N$ randomizers $\encvar{R_i}$ for $N$ variables $i$: The server first computes $\encvar{R}$ as before, then picks a fresh randomness $\delta_i \randin \Zq^*$ and sets $\encvar{R_i} \leftarrow \delta_i \cdot \encvar{R}$.\looseness=-1

\parasec{Small domain}
Malicious clients can submit non-binary ciphertexts $\encvar{z_i}$ to learn more than the cardinality.
We compute a term $\encvar{R}$ which is zero when all $z_i$s are binary, and is random otherwise.
The term $\encvar{t_i} \leftarrow \fheisin(\cpk, \encvar{z_i}, \{0, 1\})$ evaluates to 0 if $z_i \in \{0, 1\}$ and to a uniformly random element in $\Zq^*$ otherwise. Therefore, the distribution of $\encvar{R} \leftarrow \sum_{d \in D} \encvar{t_i}$ will be close to uniformly random in $\Zq$ as long as at least one non-binary $z_i$ exists in the client's query. See Appendix~\ref{ap:uniform-sum} for the exact distribution.

\parasec{Small input} 
The client sends a list of encrypted values $\encvar{x_i}$ to the server. This list represents a set as long as all elements are distinct, so the server needs to ensure that no two client elements are equal.
First, the server computes $\encvar{T}$, the product of pairwise differences of the client elements. Since these multiplications are performed in a prime group $\Zq^*$, the product will be zero if and only if two equal elements exist. Second, the server uses a uniformly random element $r \gets \Zq$ in combination with the zero detection on $\encvar{T}$ to compute the additive randomizer $\encvar{R}$.

Since the zero detection function is impractical, we provide a practical alternative protection method which deviates from our structure. The client can deterministically compute $T$, allowing us to protect the result in a multiplicative way by returning $M \gets \encvar{A} \cdot \encvar{T}$. As long as $T$ is not zero, which signifies a malicious query, the client can reverse $T$ and recover $A \gets M \cdot T^{-1}$. 

Without this protection, malicious clients are still limited to submitting a list of scalar values, with a limited size, which is equivalent to a \emph{multi-set}. Depending on the flexible matching function, allowing multi-set queries may or may not impact the security. On one hand, if the server reveals PSI-CA, allowing multi-sets may lead to the extraction of the intersection from the cardinality. On the other hand, if the server is computing F-Match, which checks $X \subseteq Y_k$, there is no difference between querying a set or a multi-set.

At first glance, $\malcheckfunc$'s  quadratic computation cost $\order(\clientsize^2)$  seems expensive. However, we target imbalanced scenarios where $\clientsize \ll \totalserversize$. In \cref{sub:asym-cost}, we discuss how this optional protection cost will be overshadowed by the cost of our PSI layer $\order(\clientsize\cdot \totalserversize)$; thus, not having any impact on our final cost.

\parasec{Next layers}
We note that the protection against malicious queries in the PSI layer extends to the matching and aggregation layers as they rely on the PSI layer to process the query. Furthermore, the randomizer $\encvar{R}$ can be applied in any layer.

\section{Matching Layer}
\label{sec:psm-layer}
Given the output of the PSI layer, the Matching layer determines whether the server set $Y_k$ is of interest to the client. The Matching layer outputs a matching status $\psmresp$. Similar to the inclusion test, $\fheisin$ in \cref{sub:fhe}, $\gamma_k$ is zero for sets of interest and a random value in $\Zq^*$ otherwise. The value $\psmresp$ can be revealed as a binary output $\lambda_k$ or passed-on to the next layer.

These matching operations can use either the small input or the small domain PSI layer protocols.
To instantiate a matching protocol, the client and the server proceed as in \cref{fig:single_psi_proto}, but plug in the desired $\serverfunc$ variant based on $f_M$.
As $\psmserverfunc$ functions have identical outputs, they share the same ${\psmclientfunc}$ method. 

We provide three matching functions $f_M$: full matching (F-Match), which determines if the client's query set is fully contained in the server set; threshold matching (Th-Match), which determines if the size of intersection exceeds a threshold; and Tversky matching (Tv-Match), which determines if the Tversky similarity between the client's and the server's set exceeds a threshold. The associated $\serverfunc$ variants are described in Algorithm~\ref{alg:single_psm_algs}.

\begin{algorithm}[tb]
    \footnotesize
    \caption{Processing matching variants.
     }
    \label{alg:single_psm_algs}
    \begin{algorithmic}
        \Function{$\tcserver{\addvariant{F}{\psmserverfunc}}$}{$\cpk,  \clientquery, Y_k$}
            \Comment Full match
            \State $\varAsList{\termisin{i}} \leftarrow \addvariant{PSI}{\serverfunc}(\cpk,  \clientquery, Y_k)$
            \State $\psmresp \leftarrow \sum_{i \in \NatNumUpTo{n}} \termisin{i}$
            \State \Return $\psmresp$
        \EndFunction
        \interalgspace

        \Function{$\tcserver{\addvariant{Th}{\psmserverfunc}}$}{$\cpk, \clientquery, Y_k, \auxdata = t_{\text{min}}$}
            \Comment Threshold
            \State $\enccardinality \leftarrow \addvariant{ePSI-CA}{\serverfunc}(\cpk,  \clientquery, Y_k)$
            \State $T \leftarrow \{t \mid t \in \Zq,  \tmin \leq t \leq \min(|Q|, |Y_k|) \}$
            \Comment Threshold to set
            \State $\psmresp \leftarrow \fheisin(\cpk, \enccardinality, T)$
            \State \Return $\psmresp$
        \EndFunction
        \interalgspace

        \Function{$\tcserver{\addvariant{Tv}{\psmserverfunc}}$}{$\cpk, \clientquery, Y_k, \auxdata = (t, \alpha, \beta)$}
            \Comment Tversky

            \State $\enccardinality \leftarrow \addvariant{ePSI-CA}{\serverfunc}(\cpk,  \clientquery, Y_k)$

            \State $(a, b, c) \leftarrow \tverskyParamTrans(\alpha, \beta, t)$

            \State $T = \{t \mid t \in \Zq,  0 \leq t \leq (a-b-c)|Y_k|\}$
            \State $\encvar{\textsf{sizeX}} \gets \left\{
              \begin{array}{@{}l@{\hskip4pt}l}
                \encvar{|Q|} & \textrm{For small input size variants} \\
                \sum_i \encvar{z_i} & \textrm{For small domain var. with $Q = \varAsList{\encvar{z_i}}$}
              \end{array}
            \right.$
            \State $\encvar{Tv} \leftarrow a \cdot \enccardinality - b \encvar{\textsf{sizeX}} - c|Y_k|$
            \State $\psmresp \leftarrow \fheisin(\cpk, \encvar{Tv}, T)$
            \State \Return $\psmresp$
        \EndFunction
        \interalgspace

        \Function{$\tcclient{{\psmclientfunc}}$}{$\csk,  M = \psmresp$}
            \Comment Reveal matching output
            \State  $\gamma_k \leftarrow \fheDec(\csk, \psmresp) $
            \State $\lambda_k \leftarrow \mathbbm{1}[\lambda_k = 0] $
            \State \Return $\lambda_k$
        \EndFunction
    \end{algorithmic}
\end{algorithm}

\para{Full matching} The F-Match variant determines if all the client elements are inside the server's set, i.e., $f_M(X, Y_k) = funcind{X \subseteq Y_k}$.  The server first computes the inclusion statuses $\termisin{i}$ by calling $\addvariant{PSI}{\serverfunc}$ (see $\addvariant{F}{\psmserverfunc}$). Recall $\termisin{i}$ is zero when $x_i \in Y_k$.
Therefore, when $X \subseteq Y_k$, the sum $\psmresp$ of all $\termisin{i}$ is zero. When an element $x_i$ is not in the server set, its inclusion status $s_i$ is uniformly random in $\Zq^*$, and therefore the sum $\lambda_k$ is also random.

The F-Match protocol has a small false-positive probability when more than one $x_i$ exists such that $x_i \notin Y_k$. Adding multiple random $\termisin{i} \in \Zq^*$ PSI responses can, incorrectly, lead to a zero sum $\psmresp$. In Appendix~\ref{ap:uniform-sum}, we bound the probability of a false-positive to $1/(q-1)$. Moreover, we bound the difference between the distribution of the sum $R$ and uniformly random over $\Zq$ to $1/{(q-1)^2}$  at all points when more than one $x_i$ is missing. The false-positive probability is zero when only one $x_i$ is missing.

Note that the F-Match protocol computes containment and not equality. Thus, hashing and comparing the client set with server's set does not work as the server would need to hash every combination of $|X|$ server set elements resulting in an exponential cost.

\para{Threshold matching} The Th-Match variant determines if the two sets have at least $\tmin$ elements in common, i.e., $f_M(X, Y_k; \tmin)=\funcind{|X \cap Y_k| \geq \tmin}$. The server first computes the encrypted cardinality $\enccardinality$ using $\addvariant{ePSI-CA}{\serverfunc}$, then evaluates the inequality $\cardinality \geq \tmin$ (see $\addvariant{Th}{\psmserverfunc}$).

Directly computing this one-sided inequality over encrypted values is costly. However, we know that
$|X\cap Y_k| \leq \min(|Q|, |Y_k|) = t_{max}$, where we bound the client set size by the size of the query.
The server evaluates the inequality $\tmin \leq \cardinality \leq \tmax$ by performing the inclusion test $\cardinality \in \{ t \mid \tmin \leq t \leq \tmax \}$ using $\fheisin$.

\para{Tversky similarity} The Tv-Match variant determines if the Tversky similarity of the two sets exceeds a threshold $t$. Formally, the protocol computes $f_M(X, Y_k; \alpha, \beta, t) = \funcind{\tverskyfull(X, Y_k) \geq t}$ where
\begin{equation*}
\tverskyfull(X, Y_k) = \frac{|X \cap Y_k|}{|X \cap Y_k| + \alpha|X - Y_k| + \beta|Y_k - X|}
\end{equation*}
is the Tversky similarity with parameters $\alpha$ and $\beta$. Computing the Tversky
similarity in this form is difficult as it requires floating-point operations. We assume $t, \alpha,$ and $\beta$ are rational, and known to both the client and the server. 
We follow the approach of Shimizu et al.~\cite{shimizu2015} and transform the inequality $\tverskyfull(X,Y) \geq t$ to 
\begin{align*}
      (t^{-1}-1+\alpha+\beta)|X \cap Y| - \alpha|X| - \beta|Y| & \geq 0 \\
      \tag{1} \Rightarrow (a, b, c) \in \Zq^3, \quad a|X \cap Y| - b|X| - c|Y| & \geq 0 \label{eq:tv-base}
\end{align*}
for appropriate integer values of $a, b, c$. The server either knows $|X|$ ($|X| = |Q|$) or can compute it ($\encvar{|X|} = \sum_i \encvar{{z_i}}$ for small domain protocols). The server also knows $|Y_k|$ and can compute $\encvar{\cardinality} = |X \cap Y_k|$ using $\addvariant{ePSI-CA}{\serverfunc}$.
Evaluating the inequality requires two steps (see $\addvariant{Tv}{\psmserverfunc}$):

\noindent \textbf{Step 1}. Transform coefficients $(t^{-1}-1+\alpha+\beta)$, $\alpha$ and $\beta$ to equivalent integer coefficients $a, b, c$. We describe this in detail in Appendix~\ref{ap:tversky}. 

\noindent \textbf{Step 2}. Evaluate the Tversky similarity inequality~\eqref{eq:tv-base}. We convert this inequality to a two-sided equation. We know $|X| \geq |X \cap Y_k|$ and $|Y_k| \geq |X \cap Y_k|$, thus 
\begin{equation*}
a|X \cap Y_k| - b|X| - c|Y_k| \leq (a-b-c)|X \cap Y_k| \leq  (a-b-c)|Y_k|.
\end{equation*}
Therefore, two sets $X$ and $Y_k$ satisfy $\tverskyfull(X, Y_k) \geq t$ iff
\begin{equation*}
  0 \leq a\enccardinality - b|X| - c|Y_k| \leq (a-b-c)|Y_k|.
\end{equation*}
The server uses an inclusion test to evaluate this inequality.

\section{Aggregation Layer}
\label{sec:aggregation-layer}
The aggregation layer combines the outputs of the matching layer, over many sets, into a single collection-wide result. We provide four aggregation functions $f_A$: naive aggregation (NA-\agg), which returns outputs as is; existential search (X-\agg), which returns whether at least one server set matched; cardinality search (CA-\agg) which returns the number of matching server sets; and retrieval (Ret-\agg) which returns the index of the $\retIndex$'th matching server set.

\Cref{fig:ms-proto} shows the structure of the aggregation protocols.
Upon receiving the query $Q$, the server runs the desired matching protocol $\psmvar$ (e.g., one from Algorithm~\ref{alg:single_psm_algs}) on each of its sets to compute the matching output $\msPSMans{j}$. The response $\msPSMans{j}$ is zero if the set $Y_j$ is interesting for the client and random otherwise. The server next runs an aggregation function $\msfunc$ that takes $N$ matching responses $\msPSMans{j}$ as input and computes the final result $A$. 
We show  how to instantiate $\msfunc$ and $\msclientfunc$ in Algorithm~\ref{alg:ms_algs}. 
If using naive aggregation (NA-\agg), the client runs $\addvariant{NA}{\msclientfunc}$. Otherwise,  the client runs ${\msclientfunc}$ to compute the result.
Finally, the server can run $\malcheckbothfunc$ on the query and apply it to the final result $A$ to protect against malicious clients.

\begin{figure}[tb]
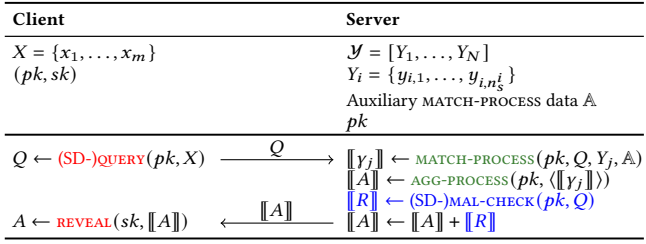

    \centering
    \footnotesize
    \begin{tabular}{@{\hskip3pt}l@{\hskip3pt}c@{\hskip3pt}l@{}}
        \toprule
        \textbf{Client} & & \textbf{Server} \\
        \midrule
        $ X = \{x_1, \ldots, x_m\} $ & & $ \collec=[Y_1, \ldots, Y_\serversetnum]  $ \\
      $ (\cpk, \csk)  $ &  & $Y_i = \{y_{i,1}, \ldots, y_{i,\serversizekth}\}$ \\
         & & Auxiliary $\psmvar$ data $\auxdata$ \\
         & & $\cpk$ \\
         \specialrule{\heavyrulewidth}{0pt}{5pt}
         $Q \gets \tcclient{\textsc{(SD-)}\queryfunc}(\cpk, X)$ &
             \diagramsend{Q } & $\msPSMans{j} \leftarrow \tcserver{\psmvar}(\cpk, Q, Y_j, \auxdata)$ \\
        &   & $\msresponse \leftarrow \tcserver{\msfunc}(\cpk,  \varAsList{\msPSMans{j}})$ \\
        & & \textcolor{blue}{$\encvar{R} \leftarrow \textsc{(SD-)}\malcheckfunc(\cpk, Q)$} \\
        $A \leftarrow \tcclient{\msclientfunc}(\csk, \msresponse)$ & \diagramrecv{\msresponse}  & $\encvar{A} \gets \encvar{A} + \textcolor{blue}{\encvar{R}}$ \\
        \bottomrule
    \end{tabular}
    \caption{Aggregation protocol structure.}
    \label{fig:ms-proto}
\end{figure}

\begin{algorithm}[tb]
    \footnotesize
    \caption{Processing $f_M$.}
    \label{alg:ms_algs}
    \begin{algorithmic}
        \Function{$\tcserver{\addvariant{NA}{\msfunc}}$}{$\cpk, \varAsList{\msPSMans{j}}$}
            \State \Return $\varAsList{\msPSMans{j}}$
        \EndFunction

        \Function{$\tcclient{\addvariant{NA}{\msclientfunc}}$}{$\csk,  M = \varAsList{\msPSMans{j}}$}
            \Comment Naive output
            \State  $\textsf{match}_j \leftarrow {\psmclientfunc}(\csk,  \msPSMans{j})$
            \State \Return $\varAsList{\textsf{match}_j }$
        \EndFunction
        \interalgspace

        \Function{$\tcserver{\addvariant{X}{\msfunc}}$}{$\cpk, \varAsList{\msPSMans{j}}$}
            \Comment At least one match
            \State $\msresponse \leftarrow \prod_{j \in \NatNumUpTo{\serversetnum}} \msPSMans{j}$
            \State \Return $\msresponse$
        \EndFunction
        \interalgspace

        \Function{$\tcserver{\addvariant{CA}{\msfunc}}$}{$\cpk, \varAsList{\msPSMans{j}}$}
            \Comment Number of matches
            \State $\msZeroAns{j} \leftarrow \fheiszero(\cpk, \msPSMans{j})$
            \State $\msresponse \leftarrow \sum_{j \in \NatNumUpTo{\serversetnum}} \msZeroAns{j}$
            \State \Return $\msresponse$
        \EndFunction
        \interalgspace

        \Function{$\tcserver{\addvariant{Ret}{\msfunc}}$}{$\cpk, \varAsList{\msPSMans{j}}, \retauxdata, \retIndex$} \Comment Retrieve data for the $\retIndex$'th match
            \State $\msZeroAns{j} \leftarrow \fheiszero(\cpk, \msPSMans{j})$
            \State $\retCtr{0} \leftarrow \encvar{0} $
            \State $\retCtr{j} \leftarrow \retCtr{j-1} + \msZeroAns{j} $
            \Comment \#matches before $j$'th set
            \State $\encvar{I_j} \leftarrow \fheiszero(\cpk, \retCtr{j} \cdot
            \msZeroAns{j} - \retIndex) $
            \Comment Create index
            \State $\msresponse \leftarrow \sum_{j \in \NatNumUpTo{\serversetnum}} \retauxdata_j  \cdot \encvar{I_j}$
            \State \Return $\msresponse$
        \EndFunction
        \interalgspace

        \Function{$\tcclient{{\msclientfunc}}$}{$\csk,  M = \msresponse$}
            \State \Return $\fheDec(\csk, \msresponse)$
        \EndFunction

    \end{algorithmic}
\end{algorithm}

\para{Naive aggregation (NA-\agg)} The naive variant runs the matching protocol on all $N$ server sets and returns $N$ results to the client, i.e., $f_A(\lambda_1, ...\lambda_N)=\lambda_1, ...\lambda_N$. This enables our framework to support many-sets when there is no need for aggregation and reduces the client's cost by computing and sending the query $Q$ only once.

\para{Existential search (X-\agg)} The existential search variant determines if at least one  server set $Y_j$ is of interest to the client. Formally, the aggregation computes $f_A(\lambda_1, ...\lambda_N)= \funcind{\exists i \, | \lambda_i=1}$.
Recall that interesting sets produce zero matching responses $\msPSMans{j}$, so a collection will have an interesting set if and only if the product of matching responses is zero (see $\addvariant{X}{\msfunc}$). 
As responses $\gamma_j$ are elements of the prime field $\Zq$, their product will never be zero without having a zero response; thus, there will be no false-positives. \looseness=-1

\para{Cardinality search (CA-\agg)} The cardinality search variant counts the number of interesting server sets $Y_j$, i.e., $f_A(\lambda_1, ..., \lambda_N) = |\{i \, | \lambda_i = 1\}|$.
This aggregation (see $\addvariant{CA}{\msfunc}$) follows the same process as ePSI-CA and uses $\fheiszero$ to turn the matching responses $\msPSMans{j}$ into binary values $b_j$ and computes their sum.

Similar to the single set PSI-CA, we can use shuffling to convert the naive aggregation into cardinality search with minimal computational overhead. This gain comes at the cost of increased communication as the protocol needs to send the $\serversetnum$ shuffled set responses to the client instead of a single encrypted cardinality.

\para{Retrieval (Ret-\agg)} The  retrieval variant returns associated data $\retauxdata_j$ of the $\retIndex$th matching server set $Y_j$. Formally, 
$f_A(\lambda_1, ...\lambda_N)=\retauxdata_j \, | (\lambda_j = 1) \land ( |\{i \, |\lambda_i = 1 \land i < j\}| = \retIndex)$. Clients use this variant when they are not concerned about whether a matching set exists, but rather about information related to this matching set -- such as an index ($\retauxdata_j = j$) for retrieving records. A good example is the matching scenario where apps want to retrieve a lot of data about the matching records. Apps would first run the Ret-\agg protocol to retrieve the index of the matching record and then follow with a PIR request to retrieve the matching set's associated data.

The Ret-\agg protocol takes an input parameter $\retIndex$ denoting that the client wants to retrieve the associated data of the $\retIndex$th matching set. The server builds an encrypted index of interesting sets in three steps (see $\addvariant{Ret}{\msfunc}$): (1) The server uses $\fheiszero$ to compute $\msZeroAns{j}$ indicating if a set is interesting. (2) The server computes a counter $\retCtr{j}$ to track how many interesting sets exist in the first $j$ sets. (3) The server combines $\msZeroAns{j}$ and $\retCtr{j}$ to compute an index $\encvar{I_j}$ where $I_j$ is 1 if $Y_j$ is the $\retIndex$th interesting set, and zero otherwise. Adding weighted $\encvar{I_j}$ values produces the result. \looseness=-1

\section{Security and Privacy}
\label{sec:priv}
In \cref{sub:formal-def} we define correctness (\cref{def:correctness}), client privacy (\cref{def:cl-priv}), and server privacy (\cref{def:sv-priv}) of PCM protocols. 
Now we discuss the security of our framework in the following 3 threat models: semi-honest, malicious client, and malicious server. Our theorems rely on properties of the HE schemes, which are all defined formally in \cref{ap:circuit-priv}. 

\parasec{Semi-honest} When both parties are semi-honest, our protocols achieve all three security properties.

\begin{theorem}
    \label{thm:semi-honest}
    Our protocols are correct, client private, and server private against semi-honest adversaries as long as the HE scheme is IND-CPA secure and circuit private.
\end{theorem}

We simulate our protocols in a real-world/ideal-world setting to prove \cref{thm:semi-honest} in \cref{ap:simulation}.
 
\parasec{Malicious server} When the server is malicious, server privacy is not applicable and we do \emph{not} provide a correctness guarantee to clients.
Malicious servers can produce arbitrary, fixed (always match/no match), or input-dependent responses to the clients without detection. 
However, servers cannot learn information about the clients' private input even when they misbehave.

\begin{theorem}
    \label{thm:client-priv}
    Our protocols provide client privacy against malicious servers, as long as the HE scheme is IND-CPA.
\end{theorem}

We prove \cref{thm:client-priv} by reducing our client privacy to our HE scheme's IND-CPA security in \cref{ap:mal-server}. The client's sole interaction is sending encrypted queries. The server cannot extract information from these encrypted queries without knowing the secret key or breaking the security of the HE scheme.

\parasec{Malicious client} When the client is malicious, client privacy is not defined. As the server receives no output in our protocols, correctness is not applicable either. Therefore, we only discuss server privacy in this setting.

\begin{theorem}
    \label{thm:server-priv}
    Our protocols provide server privacy against malicious clients if the HE scheme is IND-CPA and strongly input private.
\end{theorem}

Performing simulations on HE-based protocols in a malicious setting is very challenging as the security of the HE scheme prevents extracting effective inputs necessary for the simulation. In \cref{ap:mal-priv}, we extend the real-world/ideal-world simulation into a new notion called cipher-world, where the trusted party can extract secrets from HE keys. We show that the cipher-world provides the same server privacy guarantee as the ideal-world; then, we use this new paradigm to prove  \cref{thm:server-priv}.\\

\parasec{System-wide security} We proved the security of our framework in isolation, where there is no honest interaction outside our framework. 
In practice, when integrating our framework into a larger system, e.g., allowing journalists to contact document owners after a document search, user interactions may leak information such as the outcome of the search. Therefore, whenever protocol designers decide to integrate our framework into a larger system, they \emph{must perform a system-wide security analysis}.

\newcommand{\malCheck}{R}
\section{From Theory To Practice}
\label{sec:practical}
We first analyze the asymptotic cost of our schemes, and explain the optimizations we implement to make our schemes practical. 

\subsection{Asymptotic cost}
\label{sub:asym-cost}
Our framework is modular and supports arbitrary protocol combinations. Hence, we report the cost of layers separately. 
\cref{tab:psi-cost} summarizes the asymptotic costs of modules in our framework. 
We report the number of homomorphic additions, multiplications, and exponentiations (counting scalar-ciphertext operations as ciphertext-ciphertext operations; but excluding operations in earlier layers).
The PSI and matching layers' costs are reported for a single set ($\serversizekth$ is the size of the server's $i$'th set).
As the protection against malicious clients is optional, we report its one-time cost separately. 

As an example, we compute the cost of using existential search with full-match interest criteria with no input domain restriction. 
The client runs a single X-\agg protocol, which calls one F-\match and one small input PSI  protocol per server set.
This leads to $N$ multiplication in the aggregation layer, $\sum_i^N (\clientsize) = N\clientsize$ addition in the matching layer, and $\sum_i^N \clientsize\serversizekth = \clientsize\totalserversize$ multiplications in the PSI layer. Additionally, the server needs to perform one malicious protection per query leading to ${\clientsize}^2$ extra multiplications. The total cost will be  $N + N\clientsize + \clientsize\totalserversize + {\clientsize}^2$, but knowing that the total number of server elements ($\totalserversize$) is larger than the set number ($N$)  or the client size ($\clientsize$), we simplify the cost to $\order(\clientsize\totalserversize)$.

\newcommand{\sdMalCost}{\textcolor{blue}{|D|^*}}
\newcommand{\sdMalCostP}{\textcolor{blue}{+|D|^*}}

\begin{table}[tb]
    \caption{Asymptotic cost of our protocols. See \cref{tab:notation} for a summary of our parameter definition. SD stands for small domain protocols.
    The column \emph{Calls} denotes which protocols from earlier layers are called.
    }
    \begin{tabular}{lccccc}
        \toprule
         & Calls & Add & Mult. & Exp. &  \\ 
        \midrule
        PSI            & - & $\clientsize \serversizekth$ & $\clientsize \serversizekth$ &  &  \\
        PSI-SD         & - & & $|D|$ & \\
        ePSI-CA        & $1\times$PSI & $\clientsize$ &  &  \\
        ePSI-CA-SD     & $1\times$PSI-SD & $|D|$ &  &  \\
        \midrule
        F-\match          & $1\times$PSI & $\clientsize$ &  &  \\
        Th-\match         & $1\times$ePSI-CA & $|T|$ & $|T|$ & \\
        Tv-\match         & $1\times$ePSI-CA & $|T|$ & $|T|$ & \\
        \midrule
        NA-\agg      & $N\times$\match & & & \\
        X-\agg       & $N\times$\match & & $N$ & \\
        CA-\agg      & $N\times$\match & $N$ & & $N$ \\
        Ret-\agg     & $N\times$\match & $2N$ & $2N$ & $2N$ \\
        \midrule
        Mal protect      & - & & $\clientsize^2$ &  \\
        SD-Mal protect     & - & $|D|$ & $|D|$ &  \\
        \bottomrule
    \end{tabular}
    \label{tab:psi-cost}
\end{table}

\para{Communication} 
The client sends a query $Q$ in the PSI layer and receives a response $A$ from the aggregation layer.
The query's size only depends on the PSI layer variant and is independent from the server's input.
Small input queries are an encrypted list of client elements containing $\clientsize$ encrypted scalars while small domain queries are encrypted bit-vectors containing $|D|$ scalar elements. 
The response size depends on the aggregation method. 
Naive aggregation has to produce $N$ individual responses for $N$ sets, while other aggregation methods produce a single scalar value. This results in a total cost of $\order(|Q|)$ for non-naive and $\order(|Q|+N)$ for naive protocols which is the minimum achievable cost.

\para{Computation}
The PSI layer dominates the cost of our framework. Each PCM run calls $N$ instances of PSI leading to a cost of $\order(\clientsize \totalserversize)$ for small client input and $\order(\serversetnum\cdot |D|)$ for small input domain variants. 
Our cost is asymptotically higher than alternatives~\cite{CristofaroGT12, PinkasRTY20, KLS17} and is often dismissed as ``quadratic'' and too expensive. 
In contrast, approaches~\cite{CT09, CT10, ACT19, Pinkas0TY19, CiampiO18} with ``linear'' cost $\order(k(\clientsize + \totalserversize))$, where $10 < k < 100$ determines the protocol's false positive rate, are considered acceptable.
The extreme imbalance requirement (\reqlink{RQ.3}) leads to scenarios where $\clientsize \ll \totalserversize$ or even $\clientsize < \log(\totalserversize)$ and impacts how asymptotic costs should be interpreted. We show a how ``quadratic'' cost $\order(\clientsize \totalserversize)$ can outperform the  ``linear'' cost  $\order(k(n_c+N_s))$ in \cref{sec:eval}.

\subsection{Implementation}
Fully homomorphic encryption schemes assume unbounded multiplication depth and rely on bootstrapping, which is prohibitively expensive. Thus, we use the somewhat homomorphic BFV cryptosystem~\cite{FanV12} with a fixed multiplicative depth. 

Let $\polyDegree$ be the degree of the RLWE polynomial, $\plainMod$ be the plaintext modulus, and $\ctxMod$ be the ciphertext modulus. The polynomial degree $\polyDegree$ and ciphertext modulus $\ctxMod$ determine the multiplicative depth of the scheme. 
The plaintext modulus determines the input domain. 
We define two sets of parameters: $P_{8k}(\polyDegree=\num{8192}, \plainMod=\num{4079617})$ and $P_{32k}(\polyDegree=\num{16384}, \plainMod=\num{163841})$, and follow the Homomorphic Encryption Security Standard guidelines~\cite{Albrechtetal18} to choose $\ctxMod$s that provide 128 bits of security.
We use relinearization keys to support multiplication, and rotation keys to support some of our optimizations (see next section). Generating and communicating keys is expensive. Therefore, we assume that clients generate these keys once at setup and use them for all subsequent protocols. 
In \cref{ap:bfv-micro}, we provide more details on our parameters in \cref{table:bfv-param} and a microbenchmark of basic BFV operations and key sizes in \cref{table:bfv-cost}. 

We implement our protocols using the Go language. Our code is open-source and 1,620 lines long.\footnote{We will release the code after publication.}   
We use the Lattigo library~\cite{lattigo, MouchetTH20} for BFV operations.
Unfortunately, Lattigo does not support circuit privacy. We discuss the implications of this lack of support and possible countermeasures in \cref{ap:circuit-priv}. 
We run experiments on a machine with an Intel i7-9700 @ 3.00 GHz processor and 16 GiB of RAM. Reported numbers are single-core costs. As the costly operations are inherently parallel, we expect our scheme to scale linearly with the number of cores.

\subsection{Optimizations}
\label{sub:fhe-optimization}

We explain how we optimize our implementation and limit the multiplicative depth of (some of) our algorithms to improve efficiency. 

\para{NTT batching}
We use BFV in combination with the number-theoretic transformation
(NTT) so that a BFV ciphertext encodes a vector of $\polyDegree$
elements~\cite{SmartV14}; BFV additions and multiplications act as element-wise vector
operations. This batching enables single instruction multiple data (SIMD) operations on encrypted scalars. 
Performing operations between scalars in the same ciphertext (such as computing the sum or product of elements) requires modifying their position through rotations. Applying rotation requires rotation keys. 
Batching renders $\fheiszero$ infeasible. The exponentiation with $\plainMod - 1$ consumes a multiplicative depth of $\lg(\plainMod)$ in $\fheiszero$, so $\plainMod$ must be small. To use batching, however, the plaintext modulus $\plainMod$ must be prime and large enough that $2\polyDegree \mid \plainMod-1$.
Batching does not support small $\plainMod$s and consequently $\fheiszero$; we prioritize efficiency and only evaluate variants that do not require zero detection. The parameters we use support batching.

\para{Replication} The client's query is small with respect to the capacity of batched ciphertexts, i.e., $|Q| \ll \polyDegree$. We use two forms of replications to improve parallelization: powers and duplicates. When using a small input PSI variant, the client encodes \emph{powers} of each element $x$ (e.g., $[x, x^2, \allowbreak \ldots, x^k]$) into the query ciphertext to reduce the multiplicative depth of $\fheisin$, see the second variant in Algorithm~\ref{alg:isin}. 
Next, the client encodes $\clientCopyNum$ \emph{duplicates} of the full query (including powers, when in use) regardless of the PSI variant.

Replication is straightforward when the client is semi-honest, but impacts security when the client is malicious. 
We follow a similar process to $\malcheckfunc$ to enforce correct replication. The server computes a second randomizer $\encvar{\malCheck}$ such that $\malCheck$ will be zero when (1) all duplicates are equal and (2) for all consecutive power elements $x_i$ and $x_{i+1}$, the equality $x_{i+1} = x_1\cdot x_i$ holds. 
The server needs to compute this check only once per query.
We implemented this check and included its cost in all figures.

\section{PCM in Practice}\label{sec:eval}
To demonstrate our framework's capability, we solve the chemical similarity and document search problems. We discuss matching in mobile apps in Appendix~\ref{ap:dating}. We focus on end-to-end PCM solutions and do not evaluate single-set protocols or scenarios.

\subsection{Chemical similarity}

Recall from \cref{sec:case-studies} that chemical similarity of compounds is determined by computing and comparing molecular fingerprints~\cite{stumpfe2011, laufkotter2019, xue2001,willett1998,cereto2015,muegge2016}. We use the Tversky similarity matching algorithm, Tv-Match, to compute whether a compound in the seller's database is similar to the query compound. As fingerprints are short, we instantiate Tv-Match with the small-domain ePSI-CA-SD and represent molecular fingerprint as bit-vectors. 

We follow a popular configuration~\cite{shimizu2015} where compounds are represented by 166-bit MACCS fingerprints~\cite{durant2002} and Tversky parameters are set to $\alpha = 1, \beta = 1, t = 80\%$. Processing these raw parameters (see Algorithm~\ref{alg:single_psm_algs}) leads to the inequality
\begin{gather*}
  a, b, c = 9, 4, 4 \leftarrow \tverskyParamTrans(1, 1, 0.8)\\
  0 \leq 9\enccardinality - 4|X| - 4|Y| \leq |Y|.
\end{gather*}

We evaluate two aggregation policies. We apply X-\agg aggregation to determine whether at least one compound in the database matches and CA-\agg to count the number of matching compounds. These variants have high multiplicative depth, so we modify them to enable efficient deployment without relying on bootstrapping.

\parait{Existential search.} The X-\agg protocol applied to a collection of $N$ compounds requires a multiplicative depth of $\lg(N)$ to compute $\prod_{j \in \NatNumUpTo{N}}\msPSMans{j}$. This depth is too high for our parameters. Instead, we relax the output requirements and reduce the output as much as possible: For a fixed depth $l$, we return $\lceil N/{2^l} \rceil$ encrypted scalars $\encvar{\lambda_k} = \prod_{k2^l \leq j < (k+1)2^l}\msPSMans{j}$ to the client.
This relaxation reduces the privacy of X-\agg~-- it is less than full X-\agg but better than CA-\agg~-- at the gain of efficiency.  If this reduced privacy is unacceptable, the client just needs to choose larger HE parameters.

\parait{Cardinality search.} We use the shuffling variant of the CA-\agg protocol due to its lower multiplicative depth. When the server shuffles the $N$ matching responses, the client only learns the number of interesting compounds. Shuffling batched encrypted values is hard, so the server shuffles the compounds (server sets) as plaintext before processing the query to the same effect.

Our modifications to both aggregators increase the transfer cost from a single aggregated result to linear in the number of server sets. Yet, the X-\agg protocol sends 1 scalar value per $2^l$ compounds.

\para{Evaluation}
We evaluate our similarity search on the ChEMBL database~\cite{mendez2019chembl, gaulton2012chembl}, which is one of the largest public chemical databases with more than 2 million compounds. The ChEMBL database contains compounds in the SMILES format. We use the RDKit library~\cite{rdkit} to compute MACCS fingerprints from this format.

\Cref{fig:chem_eval} shows the cost of running our protocols with the BFV parameter set $\bfvParam{32}$. We ran 5 times the experiments with databases containing up to 256k compounds, and 3 times the larger experiments. We report the average cost. Standard errors of the mean are small, so error bars are not visible.

The server batches 128 compounds per ciphertext. Performing the PSI layer protocol ePSI-CA only requires 1 multiplication per ciphertext, while computing the binary Tversky score requires a depth of 7.
For the X-\agg protocol, we set $l=6$ and aggregate up to $64$ matching results into each scalar $\encvar{\gamma_k}$. As each ciphertext holds 32k scalars, the aggregation of up to $2M$ server sets requires 1 ciphertext.
CA-\agg transfers 1 ciphertext per 32k compounds. The cost of the extra X-\agg multiplications is insignificant compared to computing similarity, leading to similar computation costs for both protocols. \looseness=-1

The client computation of searching among 2 million compounds is less than a second and the transfer cost is 12MB for X-\agg and 378MB for CA-\agg. These protocols can be run by a thin client. The server, however, requires 3.5 hours of single-core computation.

\begin{figure}[tb]
  \centering
  \includegraphics[scale=0.5]{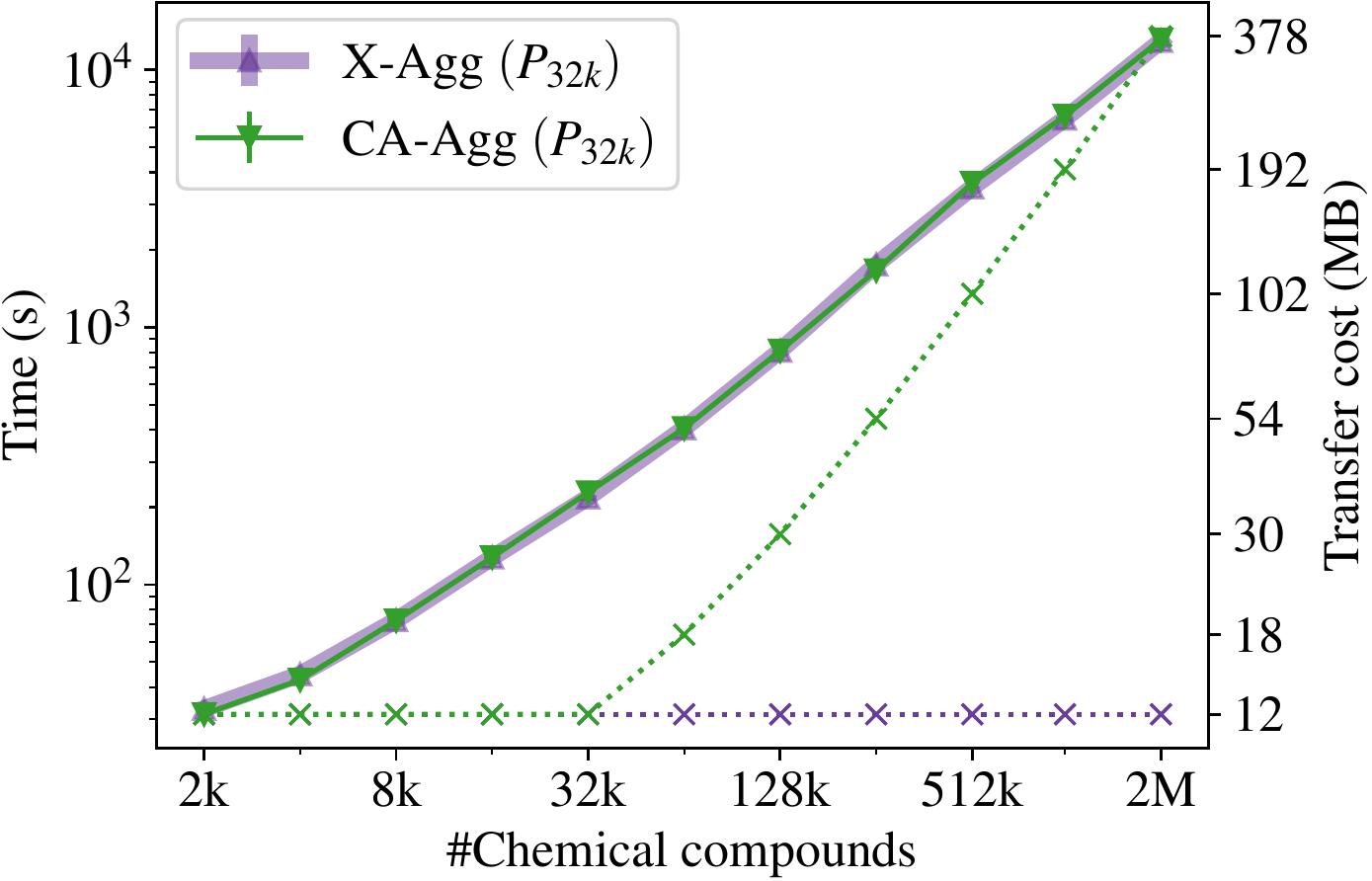}
  \caption{Computation time (solid lines) and transfer cost (dotted lines) for
    computing aggregated chemical similarity.\looseness=-1}
  \label{fig:chem_eval}
\end{figure}

\begin{figure*}[tbp]
  \centering
  \includegraphics[width=0.9\textwidth]{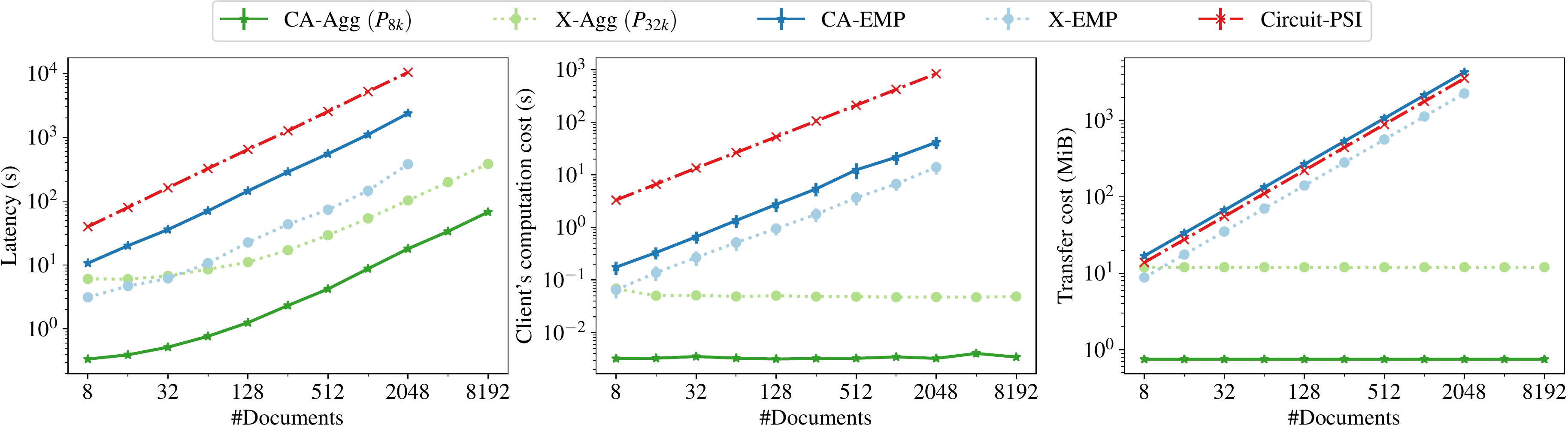}
  \caption{
    The end-to-end search latency (left), client's computation cost (middle), and communication cost (right) of document search. We limit execution to a single CPU core and enforce a bandwidth of 100 Mbps and an RTT of 100~ms on the connection.
  }
  \label{fig:doc_search}
\end{figure*}

\parasec{Comparison to ad-hoc solution}
Shimizu et al.~\cite{shimizu2015} build a custom chemical search that computes the number of matching compounds. They offer the same functionality/privacy as our CA-\agg protocol, but reveal more information than our X-\agg variant. We have identical threat models -- protect privacy against malicious adversaries while malicious servers can violate correctness. However, they use the secp192k1 curve which provides less than 100-bit of security, while we offer full 128-bit security. Moreover, their use of differential privacy requires distributional assumptions. 
Shimizu et al. report the cost of searching \emph{1.3 million compounds} as: 167 seconds of server computation, 172 seconds of client computation, and 265MB of data transfer.
Our X-\agg and CA-\agg solutions achieve higher security, lower client computation, and lower ratio of bandwidth consumed per compound, at the cost of higher server computation.\looseness=-1

\subsection{Peer-to-Peer document search}\label{sub:doc-search-eval}
To implement peer-to-peer document search, we represent queries and documents by the set of keywords they contain. A single document, represented by the keyword set $S$, is of interest to the querier if it contains all query keywords $Q$, i.e., $Q \subseteq S$. This functionality can be implemented with the full matching (F-Match) variant.

The client and the server must agree on how keywords are represented in $\Zq$. We use hash functions to do the conversion. As the search queries typically contain few keywords and the domain for searchable keywords is too large for our small domain variant, we construct F-Match with PSI with a small client input.
There are two sources of false-positive in our setup:

\para{False positive of mapping words}
The parameter $\plainMod$ (recall $\plainMod = q$) determines the input domain.
Since $\plainMod$ is small, multiple words could be mapped to the same $\Zq$ element. This can lead to F-Match claiming there is a match, even though one of the colliding keywords is absent in the server's set. Since $\plainMod$ impacts the multiplicative depth of our HE schemes, choosing a large value is impractical.

Instead of directly increasing the size of $\plainMod$ to reduce the false positive rate, the client hashes the $n_c$ query keywords with $t$ hash functions and encodes them into $t\,n_c$ scalar values, which reduces the false-positive rate to $1/{q^t}$.
When running $\addvariant{PSI}{\serverfunc}$, the server knows the corresponding hash function for each scalar value and hashes them accordingly.
Afterward, the server runs the F-Match protocol on all $t\,n_c$ PSI outputs together; 
a document matches if and only if all hashed keywords are present.

Using multiple hashes to reduce false positives does not impact privacy, as it is straightforward to simulate a query with $t$ hashes, with $t$ F-Match queries.
This modification increases the computation cost and the number of scalar values in the query by a factor of $t$. However, there is no concrete change in the communication cost as the client can still pack $tn_c$ scalars into one ciphertext.

\para{False positive of F-Match} The F-Match protocol itself has a false positive rate of $\sim 1/q$ (see \cref{sec:psm-layer}) caused by internal randomness. An easy way to reduce this FP rate is to run $r$ instances of F-Match with different randomness and reveal the $r$ responses. This repetition reduces the FP rate of a single matching response to $\sim 1/{q^r}$, while increasing the server's computation cost by a factor $r$.

\para{Aggregation} We explore two aggregation policies: existential (X-\agg) to determine whether at least one document in the collection matches; and cardinality (CA-\agg) to count matching documents. Since the multiplicative depth of F-Match is low, we can fully reduce the X-\agg output to one encrypted scalar. For the CA-\agg variant, we still use the shuffling variant for lower multiplicative depth.

\para{Evaluation}
We use the parameters $\bfvParam{8}$ for CA-\agg and $\bfvParam{32}$ for X-\agg protocols. We base our evaluation on the requirements set out in EdalatNejad et al.~\cite{DatashareNetwork}. We limit the number of keywords in each query to 8 and generate random documents of 128 keywords. We represent each keyword with two hash functions leading to false-positive rates of $2^{-44}$ for CA-\agg and $2^{-39}$ for X-\agg due to the mapping to $\Zq$. 
We need to account for false-positive errors as we run a single F-Match per document.
This error increases with the number of documents and reaches its peak at 8k documents where the probability of overestimating the cardinality is $0.2\%$ (CA-\agg) and existence is 1\% (X-\agg). Our protocols do not have false negatives.
We skip $\malcheckfunc$ from \cref{sub:base-malicious} since F-Match is not impacted by repeated keywords in the query, but still perform checks from \cref{sub:fhe-optimization} to enforce honest batching. Our use of power replication leads to a multiplicative depth of 1 in the PSI layer while F-Match only uses addition. 
We report the cost of our document search in \cref{fig:doc_search} and discuss our performance in \cref{sub:circuit-bench}.\looseness=-1

\parasec{Comparison to ad-hoc solution} 
EdalatNejad et al.~\cite{DatashareNetwork} build a document search engine that performs PSI-CA in a many-set setting and post-processes the output, in the client, to detect relevant documents.
We have identical threat models -- guarantee privacy against malicious adversaries and correctness against semi-honest servers. To enhance performance, they sacrifice privacy and leak more information than individual intersection cardinalities, yet less than vanilla PSI.
The protocol of EdalatNejad et al. has a latency of less than 1 second to search 1k documents while our framework requires $\SI{8.7}{\second}$ for CA-\agg and $\SI{54}{\second}$ for X-\agg. 
Our CA-\agg search does not reveal information about individual documents and X-\agg \emph{only reveals a single bit about the collection}. As expected, this privacy gain comes at a performance cost.\looseness=-1

\subsection{Comparison with generic solutions}
\label{sub:circuit-bench}

After comparing against ad-hoc solutions, we compare our document search against generic SMC and circuit PSI that can offer the same privacy and functionality (see \cref{tab:related-work}). See \cref{ap:circuit-bench} for implementation details. In \cref{ap:spot-bench}, we compare our document search to one of the fastest OT-based solutions, which \emph{cannot} satisfy our privacy requirements.

\parasec{Generic SMC} We use an SMC compiler, EMP tool-kit~\cite{emp-toolkit}, to build custom search engines equivalent to our X-\agg and CA-\agg searches.
The circuits we designed support many-sets.
\parasec{Circuit-PSI} We compare against Chandran et al.~\cite{ChandranGS22} which is a state-of-the-art single-set circuit-PSI protocol. To support many-set, we run one PSI per document. 
Chandran et al.\ support extending the intersection computation with arbitrary circuits allowing us to perform F-Match and then X-\agg or CA-\agg aggregation. However, we consider the cost of computing the intersection as a lower bound for the cost of search and do not extend the circuit.

\Cref{fig:doc_search} reports the latency, client's computation, and transfer cost of document search. We repeat each experiment 4 times and report the average cost and the standard error of the mean. 
Unfortunately, the EMP and circuit-PSI solutions crash on runs with 4k documents and more, hence we only show entries up to 2k documents.
The stable trend of our measurements let us think that one could easily extrapolate results for those settings from the measurements that we report.
To estimate the latency of our framework, we run it in a single process and add the expected network time (as we have one round, we consider 1 round trip time plus transfer time).\looseness=-1

Now we show that our framework supports thin clients by significantly reducing the clients' computation and communication costs. Moreover, our framework provides better latency and in \cref{ap:circuit-bench} we show that we have a competitive server cost.

\para{Latency} 
Both EMP and Circuit-PSI require many rounds of communication and are thus heavily impacted by the network's round-trip time (RTT). We assume a transatlantic RTT of \SI{100}{\milli\second}.
In this setting, framework consistently outperforms competitors. When performing an existential search (X-\agg) our framework cuts the end-to-end latency in half while our improvement factor increases to $96$x when doing a cardinality search (CA-\agg) on 1k documents.

\para{Client's computation}
Circuit-based approaches have a balanced computation load between clients and servers while we outsource almost all the computations to the server. The client's computation cost of our framework is independent of the number of the server sets and is only $\SI{50}{\milli\second}$ for the X-\agg and $\SI{5}{\milli\second}$ for the CA-\agg search.
While slower, the EMP approach is computationally efficient and only requires $\SI{3.7}{\second}$ for the existential (X-\agg) and $\SI{11.3}{\second}$  for the cardinality (CA-\agg) search for a 1k document collection. EMP's cost is still acceptable for thin clients, however using our framework leads to a saving factor of \numrange{75}{2250}x on client computation. In contrast, Chandran et al. require a prohibitive cost of $\SI{352}{}$ seconds, which is \emph{$\num{70000}$x higher than ours}.

\para{Communication} The X-\agg protocol has a constant communication cost of 12 MB while the cost of CA-\agg grows linearly with the number of documents. However, this cost is fixed to 768KB for CA-\agg  in our evaluation since we can pack up to 8k results in a $\bfvParam{8}$ ciphertext. Both EMP and Circuit-PSI have costs linear in the inputs size. The EMP search requires $\SI{1.1}{\gibi\byte}$ for X-\agg and $\SI{2.1}{\gibi\byte}$ for CA-\agg queries when searching 1k documents while Chandran et al. require $\SI{1.7}{\gibi\byte}$; thus, both approaches are prohibitively expensive.
Ultimately, our framework reduces communication by a factor of \numrange{93}{2800}x.

\section{Takeaways and Future Work}
In this work, we introduce and formalize private collection matching problems. 
Using homomorphic encryption, we build a modular framework for solving them. 
Our framework and its layers-based design simplify the construction of PCM schemes that achieve complex goals while limiting the leakage to what is required by the functionality, nothing more.
Relying on homomorphic encryption is extremely advantageous in theory. Our work shows, however, that using it in practice is challenging. 
We have overcome these challenges by using optimizations from the literature, at the cost of reduced flexibility of our framework. 
Ultimately, our framework is competitive with ad-hoc solutions and outperforms generic approaches in all three latency, communication, and computation costs, sometimes by several orders of magnitude.
As example, our framework requires 12MB of communication and less than a second of client computation to determine matching of a chemical compound against a database of 2 million compounds; or respectively 768KB and 50ms to determine whether the owner of a thousand documents has content of interest to a querier journalist.
We believe further work on homomorphic encryption schemes -- and ciphertext-based comparison methods in particular -- will allow our framework to operate in a wider range of settings. 
We hope that our work fosters the evolution of homomorphic encryption in further areas so that the community can build a wider range of privacy-preserving applications.

  \bibliographystyle{ACM-Reference-Format}
  \bibliography{sources}

  \appendix
\section{Extra material}
We provide extra materials in this section.

\parasec{Summary}
\label{ap:notation}
\label{ap:proto-summary}
\cref{tab:notation} summarizes our notation and asymptotic parameters and 
\cref{tab:proto-summary} summarizes the functionality of our protocols. \looseness=-1
\begin{table}[tb]
    \caption{Notation.}
    \begin{tabular}{@{}l@{\hskip4pt}l@{}}
        \toprule
        $q, \Zq, \Zq^*$ & A prime number, a prime ring, and a prime field.\\ 
        $\secpar$ & The security parameter. \\ 
        $a \randin A$ & Draw $a$ uniformly at random from the set $A$.\\ 
        $\NatNumUpTo{n}, \varAsList{a_i}$ & Present the set $\{1, \ldots, n\}$ and the list $[a_1,  \ldots, a_m]$.\\
        $\mathbbm{1}[E]$ & Function that returns `1' when $E$ is true and `0' otherwise.\\
        $\textsf{HE}$ & An IND-CPA circuit-private homomorphic scheme.\\
        $\cpk, \csk$& The client's public and private HE keys.\\
        $\encvar{a}$& An encryption of $a$.\\
        $\polyDegree$ & The degree of the RLWE polynomial. \\
        $\plainMod, \ctxMod$ & The plaintext and ciphertext modulus of the HE scheme.\\
        \midrule
        $X, \collec, Y_i$ & The client's set, server's collection, and server's $i$'th set.\\
        $\clientsize, \serversizekth$ & The size of client set $|X|$ and server's $i$'th set $|Y_i|$.\\
        $N, \totalserversize$ & The number of server sets and their total size $\totalserversize=\sum_i \serversizekth$.\\
        $D, T$ & The set input domain and matching threshold $[t_{min}, t_{max}]$.\\
        \midrule
        $\encvar{z_i}$  & The client's encrypted bit-vector $z_i \leftarrow  (d_i \in X)$. \\
        $Q$  & The client's query. Small input: $\varAsList{\encvar{x_i}}$, domain: $\varAsList{\encvar{z_i}}$. \\
        $\termisin{i}$ & An encrypted status determining iff $x_i \in Y_k$.\\
        $\encvar{\cardinality}$ & An encrypted cardinality of intersection $\cardinality=|X \cap Y_k|$.\\
        $\psmresp$  & A matching response. $\lambda = 0$ iff the set $Y_k$ is interesting.\\
        $\msresponse$  & An aggregated collection-wide response.\\
        $\encvar{R}$  & A term to randomize the output of malicious users. \\
        \bottomrule
    \end{tabular}
    \label{tab:notation}
  \end{table}
\newcommand{\psmFunc}{\codify{PSM}}
\begin{table}[tb]
    \caption{Summary of our protocols. We show the computed functionalities, their output range, and auxiliary input variables in the table.}
    \centering
    \begin{tabular}{@{}ll@{\hskip6pt}l@{\hskip6pt}l@{}}
        \toprule
        Protocol & Function & Range & Aux. \\
        \midrule
        PSI         & $X \cap Y$ & $\{0, 1\}^\clientsize$ & \\
        PSI-CA      & $|X \cap Y|$ & $\mathbb{Z}$ & \\
        \midrule \addlinespace[2pt]
        \multicolumn{4}{c}{Matching: $\lambda \gets f_M(X, Y)$ \vspace*{4pt}}\\
        F-Match       & $X \subseteq Y$  & $\{0, 1\}$ &  \\
        Th-Match      & $|X \cap Y| \geq t$ & $\{0, 1\}$  & $t$ \\
        Tv-Match      & $\tverskyfull(X, Y) \geq t$ & $\{0, 1\}$ & $t, \alpha, \beta$  \\
        \midrule \addlinespace[2pt]
        \multicolumn{4}{c}{Aggregation: $A \gets f_A(\lambda_1, \ldots, \lambda_N)$ \vspace*{4pt}}\\
        NA-\agg   & $\varAsList{\lambda_i}$ & $\{0, 1\}^\serversetnum$&  \\
        X-\agg    & $\exists i \, | \lambda_i = 1$ & $\{0, 1\}$ &\\
        CA-\agg   & $|\{i \,|\, \lambda_i = 1\}|$ & $\mathbb{Z}$ & \\
        Ret-\agg  & $D_j \, |  \lambda_j=1 \land |\{i \,|\, \lambda_i=1  \land i \in \NatNumUpTo{j} \}|= \kappa  $  & $\mathbb{Z}$ & $ D,  \kappa$  \\
        \bottomrule
    \end{tabular}
    \label{tab:proto-summary}
\end{table}

\parasec{Tversky similarity}
\label{ap:tversky}
Algorithm~\ref{alg:Tversky-param-proccess} shows how to process rational Tversky parameters to enable computing similarity with modular arithmetic.

\newcommand{\torational}{\textsf{ToRational}}
\begin{algorithm}[tb]
    \footnotesize
    \caption{Process Tversky parameters $t, \alpha$, and $\beta$ to compute integer coefficients $(a, b, c)$.}
    \label{alg:Tversky-param-proccess}
    \begin{algorithmic}
        \Function{$\tverskyParamTrans$}{$\alpha, \beta, t$}
        \State $(\alpha_1, \alpha_2) \leftarrow \torational(\alpha)$
        \Comment $\alpha = \alpha_1/ \alpha_2$
        \State $(\beta_1, \beta) \leftarrow \torational(\beta)$
        \Comment $\beta = \beta_1/ \beta$
        \State $(t_1, t_2) \leftarrow \torational(t)$
        \Comment $t = t_1/ t_2$
        \State $l \leftarrow \textsf{LCM}(t_1, \alpha_2, \beta_2)$
        \State $a \leftarrow  l\cdot (t^{-1}-1+\alpha+\beta)$
        \State $b \leftarrow l\cdot \alpha$
        \State $c \leftarrow l\cdot \beta$
        \State $g \leftarrow \textsf{GCD}(a, b, c)$
        \State \Return  $(a/g, b/g, c/g)$
    \EndFunction
    \end{algorithmic}
\end{algorithm}

\parasec{BFV parameters}
\label{ap:bfv-micro}
We report full details of our BFV parameters including their supported multiplicative depth in \cref{table:bfv-param}. Next we provide a microbenchmark of basic operations and key sizes in \cref{table:bfv-cost}.
It is possible to reduce the size of the rotation keys in exchange for more costly rotation operations. We report key sizes that provide a balanced computation/communication trade-off. \

\begin{table}[tb]
    \centering
    \caption{BFV parameters with 128-bit security}
    \begin{tabular}{lcccc}
        \toprule
         & $\polyDegree$ & $\plainMod$ & $\lg(\ctxMod)$ & Mult. depth\\ 
        \midrule
        $\bfvParam{8}$ & 8192 & 4079617 &  218-bit & 2\\
        $\bfvParam{16}$ & 16384 & 163841 & 438-bit & 7\\
        $\bfvParam{32}$ & 32768 & 786433 & 880-bit & 16\\
        \bottomrule
    \end{tabular}
    \label{table:bfv-param}
\end{table}
\begin{table}[tb]
    \centering
    \caption{Cost of basic BFV operations.}
    \begin{tabular}{lccc}
        \toprule
         & $\bfvParam{8}$ & $\bfvParam{16}$ & $\bfvParam{32}$ \\ 
        \midrule
        Addition ($\mu$s)  & 29 & 115 & 530 \\     
        Multiplication (ms)  & 7.3 & 36.3 & 182\\   
        Plaintext mult. (ms)  & 0.97 & 4.13 & 18.3 \\
        Rotation by 1 (ms)  & 2.16 & 10.8 & 57 \\
        \midrule
        Ciphertext (KB) & 384 & 1536 & 6144 \\
        Public key (KB) & 512 & 2048 & 7680 \\
        Relinearization key (MB) & 3 & 12 & 60 \\
        Rotation key (MB) & 22 & 96 & 510 \\
        \bottomrule
    \end{tabular}
    \label{table:bfv-cost}
\end{table}
\section{Sum of Random $\Zq^*$ Elements}
\label{ap:uniform-sum}
In \cref{sec:psm-layer}, we argued that the distribution of the sum $s = \sum_{i=1}^k x_i$ of $k$ random $x_i \randin \Zq^*$ elements is close to uniform when $q$ is prime. Now we prove that the probability of the sum being zero is bounded by $1/(q-1)$ and the difference between the probability of the sum being zero vs a non-zero value $a$ is at most $1/{(q-1)^2}$ when $k$ is larger than one.

\newcommand{\zqdepth}[2]{{#1}^{[#2]}}

Let $\zqdepth{z}{k}$ be the probability that the of sum of $k$ elements from $\Zq^*$ is zero and $\zqdepth{p_a}{k}$ be the probability that the of sum is a non-zero value $a$. When $k = 1$, e.g.,
we sum one element, these probabilities are $\zqdepth{z}{1} = 0 $ and $\zqdepth{p_a}{1} = 1/(q-1)$. The $k$ elements $x_i$ are independent from each other so we choose the value of the last element and recursively compute the probability by using the distribution of $k-1$ elements.
\begin{align*}
    \zqdepth{z}{k} &= \sum_{i=1}^{q-1} \zqdepth{p_a}{k-1}/(q-1) = \zqdepth{p_a}{k-1}\\
    \zqdepth{p_a}{k} &= \zqdepth{z}{k-1}/(q-1 ) + \sum_{i \in \Zq^* -\{a\}} \zqdepth{p_{a-i}}{k-1}/(q-1) \\
    &= \zqdepth{z}{k-1}/(q-1 ) + (q-2)\cdot\zqdepth{p_a}{k-1}/(q-1)\\
    &= \zqdepth{p_a}{k-1} + (\zqdepth{z}{k-1} - \zqdepth{p_a}{k-1})/(q-1)\\
\end{align*}

The probability gap of $\zqdepth{z}{k}$ and $\zqdepth{p_a}{k}$ gets narrower as the number of elements $k$ increases as:
\begin{align*}
    \zqdepth{p_a}{k} - \zqdepth{z}{k} &= \zqdepth{p_a}{k-1} + (\zqdepth{z}{k-1} - \zqdepth{p_a}{k-1})/(q-1) - \zqdepth{p_a}{k-1} \\
    &=  (\zqdepth{z}{k-1} - \zqdepth{p_a}{k-1})/(q-1)
\end{align*}

The highest probability of a zero sum happens when summing two random elements and the probability is bound by $\zqdepth{z}{2} = \zqdepth{p_a}{1}= 1/(q-1)$. Similarly the highest probability difference happens when $k=2$ and is bound by $\zqdepth{z}{2} - \zqdepth{p_a}{2} = (\zqdepth{p_a}{1} - \zqdepth{z}{1})/(q-1) = 1/{(q-1)^2}$.
This means that the probability of false-positive in our approach is bounded by $1/(q-1)$ while the probability of distinguishing the number of non-zero elements (when at least one non-zero element is present) is $1/{(q-1)^2}$.

\newcommand{\simxydomain}{{X, \collec, \secpar}}
\newcommand{\simcircuitp}{\sims{sip}}
\newcommand{\serverresp}{R}
\newcommand{\clientoutput}{A}
\newcommand{\oracle}[1]{\mathbb{O}_{\textrm{#1}}}

\section{Privacy Proof}
\label{ap:priv}
We prove the security and privacy of our framework.
We throughout assume that honest users do not have any interaction outside our framework that can leak information.

\parasec{Roadmap}
We start by formally defining security properties of HE schemes such as IND-CPA security, circuit privacy, and strong input privacy in \cref{ap:circuit-priv}. 
Next, we address our framework's security in a semi-honest setting and use real-world/ideal-world simulation to prove \cref{thm:semi-honest} in \cref{ap:simulation}.
We study malicious servers in \cref{ap:mal-server} and provide a tight reduction from our client privacy to the semantic security of our HE scheme to prove \cref{thm:client-priv}.
Finally, we extend the notion of real-world/ideal-world simulation to a new paradigm called cipher-world, which provides better support for simulating HE protocols in a malicious setting. We use this new notion to prove \cref{thm:server-priv} and show that our protocols achieve server privacy against malicious clients in \cref{ap:mal-priv}.

\subsection{Security properties of HE schemes}
\label{ap:circuit-priv}
We formally define the security properties of HE schemes in this section. We start with semantic security (IND-CPA). Next, we discuss noise in HE schemes and define circuit privacy, which ensures that the noise contained in ciphertexts does not leak information about the computation performed on them. Finally, we extend circuit privacy, which only applies in a semi-honest setting, to a malicious setting and introduce strong input privacy.

\begin{definition}[IND-CPA]
  \label{def:ind-cpa}
  An encryption scheme is indistinguishable against chosen plaintext attacks if  no PPT adversary $\adv$ exists such that:
  \[
    \textrm{Pr}\left[
    \begin{array}{c}
      p \gets \fheParamGen(1^{\ell}) \\
      (\cpk, \csk) \gets \fheKeyGen(p) \\
      (\textsf{st}, m_0, m_1) \gets \adv(\cpk) \\
      b \randin \{0, 1\} \\
      c_b \gets \fheEnc(\cpk, m_b) \\
      b' \gets \adv(\textsf{st}, c_b) \\
    \end{array}
    : 
    \begin{array}{c}
      b = b'
    \end{array}
    \right]  > \frac{1}{2} + \epsilon
  \]
\end{definition}

Each BFV ciphertext $c$ contains noise. The amount of noise increases with each operation and can be measured. As a result, a ciphertext $c$ contains more information than its decrypted value $p \leftarrow \fheDec(\csk, c)$. If the noise grows beyond the HE parameter's noise budget, then the decryption fails $\bot \leftarrow \fheDec(\csk, c)$. Informally, circuit privacy states that the ciphertext $c$ should not reveal any information about the computation performed on $c$ beyond the decrypted result $p$. We follow the definition of Castro et al.~\cite{CastroJV20}:\looseness=-1

\begin{definition}[Circuit privacy]
    \label{def:circuit-priv}
    Let HE be a leveled homomorphic encryption scheme and let
    \begin{align*}
      \param & \gets \fheParamGen(q) \\
      (\csk, \cpk) & \gets \fheKeyGen(\param) \\
      c_i & \gets \fheEnc(\csk, m_i) \\
      M & \gets f(m_1, \ldots, m_n, p_1, \ldots, p_k) 
    \end{align*}
    be an (honestly) generated key pair, ciphertexts, and output of the computation. 
    The scheme HE is $\epsilon$-circuit private if a PPT simulator $\sims{}$ exists such that for all functions $f$ of depth $l \leq L$  all PPT distinguisher algorithms $\mathcal{D}$ are bounded by
    \begin{multline*}
      \Big|\textrm{Pr}\big[\mathcal{D}\left(\fheeval\left(\cpk, f, \varAsList{c_i}, \varAsList{p_i}\right), \csk, \cpk, \varAsList{c_i})\right) = 1\big] - \allowbreak \\
      \textrm{Pr}\big[\mathcal{D}\left(\sims{}\left(\csk, \cpk, M\right), \csk, \cpk, \varAsList{c_i})\right) = 1\big] \Big|  \leq \epsilon\text{.}
    \end{multline*}
\end{definition} 

The circuit privacy's definition focuses on the semi-honest setting and requires honest generation of HE keys and ciphertexts. We extend this definition by (1) removing the honest generation requirement, which extends the property to the malicious setting, and (2) relaxing privacy by revealing the functionality $f$ and only hiding the evaluator's private data $\varAsList{p_i}$.

\begin{definition}[Strong input privacy]
  \label{def:str-input-priv}
  A leveled homomorphic encryption scheme HE is $\epsilon$-strong input private if a simulator $\simcircuitp$ exists such that all PPT adversaries $\adv$ are bounded by:
\begin{multline*}
  \textrm{Pr}\left[
  \begin{array}{c}
    (\textsf{st}, \cpk, \csk) \gets \adv(1^\ell) \\
    (\textsf{st}, f, \varAsList{c_i}, \varAsList{p_i}) \gets \adv(\textsf{st}) \\
    b \randin \{0, 1\} \\
    m_i \gets \fheDec(\csk, c_i) \\
    M \gets f(\varAsList{m_i}, \varAsList{p_i}) \\
    {a_0} \gets \textsf{HE.Eval}(\cpk, f, \varAsList{c_i}, \varAsList{p_i}) \\
    {a_1} \gets \simcircuitp(\csk, \cpk, f, \varAsList{c_i}, M) \\
    b' \gets \adv(\textsf{st}, {a_b}) \\
  \end{array}
  : 
  \begin{array}{c}
    b = b'
  \end{array}
  \right] \leq \frac{1}{2} + \epsilon
\end{multline*}
\end{definition} 

\parasec{Ciphertext indistinguishability}
When simulating our protocols, in the next section, the simulator generates ciphertexts. As part of our proof, we need to show that these ciphertexts are indistinguishable from the output of our protocols. Hence, we discuss when a distinguisher $\mathcal{D}$ can distinguish two ciphertexts $c$ and $c'$ in three settings:

\para{Known public key} We first consider the case where the distinguisher only knows the public key $(\cpk)$. This scenario directly follows from the IND-CPA property. As long as the HE scheme is \emph{IND-CPA secure} and ciphertexts have the same size $|c| = |c'|$, then the distinguisher $\mathcal{D}(\cpk, c, c')$ has a negligible chance.

\para{Semi-honest with known secret key} Next, we consider the case where the distinguisher knows both the secret and public keys $(\csk, \cpk)$ in the semi-honest setting, i.e., keys and ciphertexts are honestly generated. In this scenario, the distinguisher can decrypt ciphertexts so IND-CPA is not enough. Decrypting a ciphertext $p \gets \fheDec(\csk, c)$ results in a plaintext $p$ and a measurable noise $\epsilon$.
This transforms ciphertext indistinguishability to showing two statements: both decrypted plaintexts and ciphertext noises are indistinguishable.
While comparing decrypted values $p$ is straightforward, we rely on circuit privacy to ensure noises are indistinguishable. As long as the HE scheme is \emph{IND-CPA secure} and \emph{circuit-private}, and  $p \gets \fheDec(\csk, c) \land p' \gets \fheDec(\csk, c') \land p \compequiv p'$, then the distinguisher $\mathcal{D}(\csk, \cpk, c, c')$ has a negligible chance in a semi-honest setting.

\para{Malicious with known secret key} The distinguisher knows both the secret and public keys $(\csk, \cpk)$ in a malicious setting. We require our scheme to be \emph{IND-CPA secure} and \emph{strongly input private} in this setting and use the simulator of strong input privacy $(\simcircuitp)$ to produce indistinguishable ciphertexts.

\parasec{Lack of circuit privacy} The Lattigo library does not provide circuit privacy. This is not surprising since other popular HE libraries such as Microsoft Seal~\cite{sealcrypto} do not provide circuit privacy either.
While there are possible mitigations in the semi-honest setting such as noise-flooding~\cite{Gentry09}  (alternatively called noise smudging) to achieve circuit privacy, there is no known mechanism for the malicious setting (i.e., no HE scheme achieves strong input privacy). Noise-flooding is only proven private in a semi-honest scenario where keys are generated honestly and the computation is guaranteed to start on freshly encrypted ciphertexts. 
Beyond flooding, there is a new line of work~\cite{CastroJV20} that uses careful parameter selection in RNS or DCRT representation of ciphertexts to achieve lightweight circuit privacy. We hope this approach will be adopted by HE libraries. 

Our implementation does not add extra defense mechanisms to prevent possible leakage from noise, due to the extra cost associated with these defenses. 
Despite this leakage, practical attacks using noise are limited.
The complexity and depth of our functions make extracting information from this noise more challenging; especially since the size of the server's private data, $\totalserversize$, is significantly larger than the capacity of the noise for storing information.

\subsection{Semi-honest security}
\label{ap:simulation}
We use real-world/ideal-world simulation to prove the security of our PCM protocols in the semi-honest setting (see \cref{thm:semi-honest}). Our malicious protection functions, \malcheckbothfunc, have no impact on honest execution as they only add $\encvar{0}$ to the result when the query is generated honestly. Therefore, we ignore these functions in this section. Our framework can be instantiated to support different functionalities, but they share a similar structure, which enables us to write a single proof that is customizable depending on the protocol variation. 
As seen in \cref{def:PCM}, PCM is a two-party interaction that computes 
\[\left(f_c(X, \collec), \bot\right) \leftarrow PCM_{f_M, f_A}(X, \collec)\]
where $f_c(X, \collec) = f_A\left(f_M\left(X, Y_1\right), \ldots, f_M\left(X, Y_N\right)\right)$ is a \emph{deterministic} function selected from the \cref{tab:proto-summary}.
To reason about the properties the matching layer, we allow $f_A$ to be equal to the identity function. To reason about the PSI layer, we allow $f_c$ to have the natural PSI and PSI-CA definition.

Our semi-honest scenario has deterministic output functions. Therefore, we can use the simpler formulation of security in Lindell~\cite{Lindell17} which requires schemes to satisfy two properties to be secure. \emph{Correctness}: the output of parties is correct. \emph{Privacy}: the view of parties can be separately simulated as follows:
\begin{align*}
  \left\{\sims{C}\left(1^{\secpar}, X, f_c\left(X, \collec\right)\right)\right\}_\simxydomain &\compequiv \left\{\simview{client}\left(\secpar, X, \collec\right)\right\}_\simxydomain,\\
  \left\{\sims{S}\left(1^{\secpar}, \collec, \bot\right)\right\}_\simxydomain &\compequiv \left\{\simview{server}\left(\secpar, X, \collec\right)\right\}_\simxydomain
\end{align*}
where $\secpar$ is the security parameter, $X$ is the client input, and $\collec$ is the server input. We omit $\secpar$ in the rest of this section. We assume that the client honestly generates the key pair  $(\cpk, \csk) \leftarrow \fheKeyGen(\param)$ and sends the public key $\cpk$ to the server before running the protocol.

\parasec{Correctness} In a semi-honest scenario where both parties follow the protocol specification, showing that \cref{alg:single_psi_algs,alg:single_psi_algs_sd,alg:single_psm_algs,alg:ms_algs} compute the functionality described in \cref{tab:proto-summary} is straightforward math. We have described the intuition behind these algorithms in \cref{sec:base-layer,sec:psm-layer,sec:aggregation-layer}, so we do not repeat the argument here.

\parasec{Server privacy} 
We simulate the view of clients to ensure server privacy. Let the view of the client be
\[\simview{client}\left(X, \collec\right) = (X, \text{rnd}, Q, \serverresp, \clientoutput = f_c(X, \collec))\]
where $\text{rnd}$ is the internal random tape, $Q$ is the encrypted query, $\serverresp$ is the server's encrypted response, and $\clientoutput$ is the client's output. 

All our protocols start with clients sending an encrypted query $Q$ to the server and receiving an encrypted response $\serverresp$. Afterward, clients apply the appropriate reveal function on $\serverresp$ to compute the output $\clientoutput \leftarrow {\clientfunc}(\csk, \serverresp)$.
Correctness ensures that the output $\clientoutput$  computed by the reveal algorithm is equal to the expected output $f_c(X, \collec)$ summarized in \cref{tab:proto-summary}.

Now we build a simulator $\sims{C}$ that given the input $(X, \text{rnd}, A = f_c(X, \collec))$ simulates the clients view as follows:

\begin{enumerate}
  \item The simulator uses $Q' \gets \querybothfunc(\cpk, X)$ to compute a query depending on the domain size with $\text{rnd}$ as internal randomness.
  \item We categorize the client's output based on the range of $f_c(X, \collec)$ into 3 groups: 
  \[
    \clientoutput \in
      \left\{\begin{array}{@{\hskip2pt}l@{\hskip12pt}l}
        \{0, 1\} & \textrm{For: F-\match, Th-\match, Tv-\match, X-\agg} \\
        \{0, 1\}^k & \textrm{For: PSI, NA-\agg}\\
        \mathbb{Z} & \textrm{For: PSI-CA, CA-\agg, Ret-\agg}\\
      \end{array}\right.
  \]

   Depending on the range category, the simulator computes the server response $\serverresp'$ as follows using the circuit-privacy simulator to ensure equivalent noise levels:
    \[
      \serverresp' \gets
        \left\{\begin{array}{@{\hskip2pt}l@{\hskip12pt}l}
          \fheEnc(\cpk, r \cdot \clientoutput) & \clientoutput \in \{0, 1\} \\
          \varAsList{\fheEnc(\cpk, r_j \cdot \clientoutput[j])}  & \clientoutput \in \{0, 1\}^k\\
          \fheEnc(\cpk, A) & \clientoutput \in \mathbb{Z}\\
        \end{array}\right.
    \]
    where $r \randin \Zq^* $ and $r_j \randin \Zq^*$ are random values.\footnote{When evaluating F-\match, the simulator chooses random variables from $\Zq$ instead of $\Zq^*$ to ensure the same false positive probability as the real-world execution.\looseness=-1}
  \item The simulator returns $(X, \text{rnd}, Q', \serverresp', \clientoutput)$.
\end{enumerate}

Now we show that the distribution returned from the simulator $\sims{C}$ is indistinguishable from $\simview{client}$. The three variables $X, A, \text{rnd}$ are taken from the input and are guaranteed to have the same distribution between the simulation and the view. Therefore, we only need to show that the joint distribution for the query and the server response are indistinguishable -- conditional on the common  $X, A,$ and $\text{rnd}$ variables. 

Both the query $Q$ and response $\serverresp$ are ciphertexts.
As discussed in `ciphertext indistinguishability' in \cref{ap:circuit-priv}, we are in a semi-honest scenario where the distinguisher knows both the public and private keys $(\cpk, \csk)$.
Our HE scheme is both IND-CPA and circuit private so we only need to show that these ciphertext pairs,  $(Q, Q')$ and $(\serverresp, \serverresp')$, decrypt to the same value.
Queries are computed following $Q \gets \querybothfunc(\cpk, X)$. The query function only encrypts the input set $X$ as the query so both $Q$ and $Q'$ should decrypt to the same value.
The server's response $\serverresp$ depends on the client's output $\clientoutput$ and this relation is specified in our $\clientfunc$ functions. Our correctness property ensures that the following relation between $\serverresp$ and $\clientoutput$ holds in the real-world:

\[
  \serverresp \gets
    \left\{\begin{array}{@{\hskip2pt}l@{\hskip6pt}l@{\hskip12pt}l}
      
          \encvar{b} & b = \left\{\begin{array}{@{\hskip2pt}l@{\hskip12pt}l}
            0 & \clientoutput = 0 \\ 
            \randin \Zq^* & \clientoutput = 1 \\ 
          \end{array}\right| 
          & \clientoutput \in \{0, 1\}  \\

          \varAsList{\encvar{b_j}} &  b_j = \left\{\begin{array}{@{\hskip2pt}l@{\hskip12pt}l}
            0 & \clientoutput[j] = 0 \\ 
            \randin \Zq^* & \clientoutput[j] = 1 \\ 
          \end{array}\right|
          & \clientoutput \in \{0, 1\}^k  \\

      \encvar{\clientoutput} & & \clientoutput \in \mathbb{Z}\\
    \end{array}\right. \text{.}
\]
It is straightforward to see that the underlying plaintext of $\serverresp'$, produced in step (2) by $\sims{C}$, follows the same distribution.
\hfill\qed

\parasec{Client privacy} 
We simulate the view of servers to ensure client privacy. Simulating the server's view is considerably simpler than the client's as the server only observes ciphertexts encrypted under the client's key.  In this simulation, the adversary is a semi-honest server and the distinguisher does not have access to the client's secret key. Moreover, our HE scheme is IND-CPA, which simplifies ciphertext indistinguishability to ensuring $|c| = |c'|$ (see \cref{ap:circuit-priv}).

Let the server's view be
\(\simview{server} = (\collec, \text{rnd}, Q, \serverresp)\). 
We build a simulator $\sims{S}$ that given the input $(\collec, \text{rnd})$ proceeds as follows:
\begin{enumerate}
  \item The simulator chooses two random variables $r_1$ and $r_2$ with the same size as the query and the server response then encrypts them.
  \[Q' \gets \fheEnc(\cpk, r_1), \serverresp' \gets \fheEnc(\cpk, r_2)\]
  \item The simulators return $(\collec, \text{rnd}, Q', \serverresp')$.
\end{enumerate}

The variables $\collec$ and $\text{rnd}$ are taken from the input and have the same distribution between the simulation and the view. The two variables $Q'$ and $\serverresp'$ are encrypted under the client's key. Without the knowledge of the secret key $Q \compequiv Q'$ and $\serverresp \compequiv \serverresp'$ hold independent of their content. 
\hfill\qed

\subsection{Malicious server}
\label{ap:mal-server}
Now, we study the setting where the server is malicious. In this setting, the client is honest and server privacy is not applicable. Our framework provides \emph{no correctness guarantee} in this setting. The malicious server can enforce a corrupted response that depends on the client input or fix the outcome independent of the query. Despite the lack of correctness guarantee, malicious servers cannot learn any information about the client's private data. We restate client privacy (\cref{def:cl-priv}) more formally, then prove \cref{thm:client-priv} by giving a tight reduction from our scheme's client privacy to the IND-CPA security of our HE encryption scheme. During this proof, we assume that the client's input fits into a single batched ciphertext, which holds throughout our evaluation. Adjusting the proof to support input encrypted in $k$ ciphertexts is straightforward and gives a $k$-fold advantage to the adversary.

\begin{definition*}[Client privacy]
  A PCM protocol is client private if no PPT adversary $\adv$ exists such that:
  \[
    \textrm{Pr}\left[
    \begin{array}{c}
      p \gets \fheParamGen(1^{\ell}) \\
      (\cpk, \csk) \gets \fheKeyGen(p) \\
      (\textsf{st}, X_0, X_1) \gets \adv(\cpk) \\
      b \randin \{0, 1\} \\
      {Q_b} \gets \querybothfunc(\cpk,X_b) \\
      b' \gets \adv(\textsf{st}, Q_b) \\
    \end{array}
    : 
    \begin{array}{c}
      |X_0| = |X_1|\\
      b = b'
    \end{array}
    \right] > \frac{1}{2} + \epsilon \text{.}
  \]
\end{definition*}

\begin{proof}
  Let $\adv$ be an adversary that can break the client privacy property with a non-negligible probability $\epsilon$. We build a new adversary $\adv'$ that can break the IND-CPA security (\cref{def:ind-cpa}) of our HE scheme with the same probability $\epsilon$. 

  \begin{enumerate}
    \item The adversary $\adv'$ starts an IND-CPA challenge and receives a public key $\adv'(\cpk)$.
    \item The adversary $\adv'$ calls $(\textsf{st}, X_0, X_1) \gets \adv(\cpk)$.
    \item Depending on the domain size, $\adv'$ converts sets $X_a$ to $m_a = \varAsList{x_i}$ or $m_a = \varAsList{z_i}$ following the logic of \querybothfunc. 
    \item The adversary $\adv'$ continues the IND-CPA challenge with $(\textsf{st}, m_0, m_1)$ and receives  $\adv'(\textsf{st}, c_b)$.
    \item The adversary $\adv'$ passes the challenge $b' \gets \adv(\textsf{st}, Q_b = c_b)$ to $\adv$ and returns $b'$ as the output.
  \end{enumerate}

  We show that the adversary $\adv$ cannot distinguish interaction with $\adv'$ from our protocol. The adversary $\adv$ has two interactions in the client privacy challenge. The first interaction is getting a fresh public key $\cpk$ which is the same between the IND-CPA and client privacy challenges. In the second interaction, $\adv$ receives a query $Q_b$ produced by $\querybothfunc$. We know that $\querybothfunc$ consists of two steps: convert $X$ into $p = \varAsList{x_i}$ or $p = \varAsList{z_i}$ which $\adv'$ performs in (3) and encrypting $p$ with $\cpk$ which is performed as part of IND-CPA challenge. This ensures that $c_b$ is computed in the same way as $Q_b$ and follows the same distribution.
  The adversary $\adv'$ runs one instance of client privacy with $\adv$ and succeeds the IND-CPA as long as the $\adv$ succeeds leading to the same $\epsilon$ advantage.
\end{proof}

\newcommand{\ownernoisectx}{{c_e}}
\newcommand{\ownernoiseptx}{e}
\newcommand{\owneroutputctx}{{c_{\clientoutput}}}
\newcommand{\unboundedoracle}{\textsf{CIPHER}}
\newcommand{\leakage}{\mathcal{L}}

\subsection{Malicious clients}
\label{ap:mal-priv}
When addressing malicious clients, client privacy is irrelevant and correctness does not apply either as the server has no output. Therefore, we only need to address our protocols' server privacy. Unfortunately, direct application of real-world/ideal-world simulation on our protocols in the malicious setting is not possible.
The biggest challenges in simulation proofs in the malicious setting is that the adversary is not required to use its input and random tape during the execution. Therefore, the simulation needs to extract the effective input that the malicious adversary uses to determine the corresponding output.
In our protocols, the server only receives an encrypted query from the client and there is no further interaction, such as having ROM calls in the client, to provide any extraction opportunity.
Therefore, the semantic security of our HE scheme prevents extracting the effective inputs. 
To address this challenge, we adapt the ideal-world paradigm and introduce a new notion for simulating HE protocols called cipher-world. Cipher-world is inspired by trapdoor commitments and allows the trusted party to decrypt the query and compute the ideal functionality.
We first prove that cipher-world provides the same server privacy guarantee as the ideal-world then use it to prove our framework server-private.

\subsection*{Cipher world}
We extend the notion of real-world/ideal-world simulation for two-party HE-based schemes in which exactly one party, called the owner, holds a pair of HE keys while the other party, called the evaluator, performs computation in the encrypted domain. This extension is \emph{one-sided} and only addresses the privacy of the evaluator (i.e., the server in our protocols). We call this extension cipher-world.

Cipher-world assumes that the trusted party defined in the ideal-world accepts encrypted inputs, is \emph{computationally unbounded}, and can extract the secret key $\csk$ of the owner from its public key $\cpk$ -- breaking the security of the HE scheme in the process.
The trusted party uses this secret key to extract the effective input $X$ of malicious owners.
Unlike the ideal-world, there is no guarantee that the decrypted input of the cipher oracle is valid and follows the input restrictions. The cipher oracle first verifies input restrictions and outputs $\bot$ if any check fails; otherwise, it computes and outputs $f(X, Y)$.  Additionally, the cipher-world reveals the secret key $sk$ to the simulator. This key is solely used for the purpose of simulating ciphertext noise.
\Cref{fig:cipher-world} shows the structure of the cipher-world and how it compares to the standard ideal-world setting.
We define both ideal-world and cipher-world oracles below to highlight their differences:
\begin{gather*}
  f(X, Y) \gets \textsf{IDEAL}(\secpar, X, Y)\\
  (\csk,   f(X, Y)/ \bot) \gets \unboundedoracle(\secpar, (\cpk, c = \encvar{X}), Y)\text{.}
\end{gather*}
Recall that here $Y$ contains the evaluator's private input.

\begin{figure}[tb]
  \centering
  \includegraphics[width=1\linewidth]{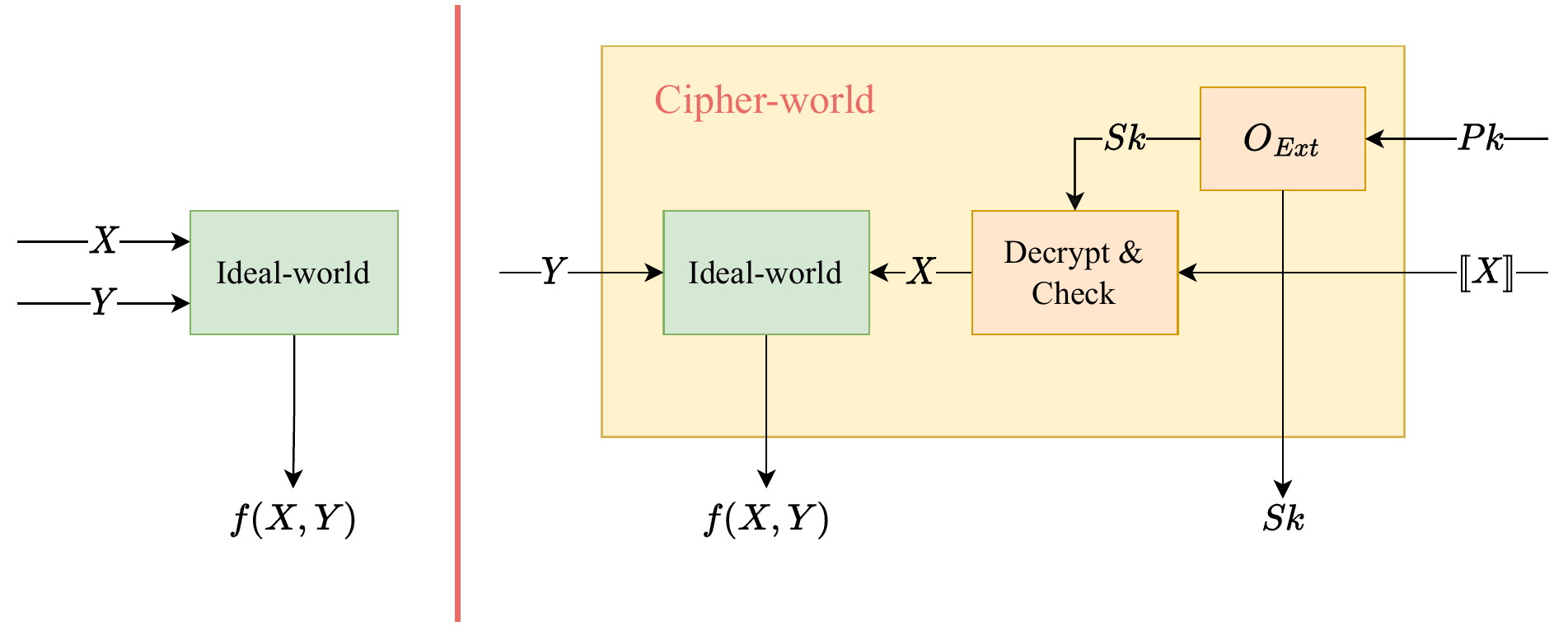}
  \caption{Structure of the cipher-world and its difference with the ideal-world.}
  \label{fig:cipher-world}
\end{figure}

\begin{theorem}
  \label{thm:evaluator-priv}
  The cipher-world provides the same privacy guarantee for the evaluator's private information as the ideal-world.
\end{theorem}

\begin{proof}
  It is clear that owners can learn more information from interacting with cipher oracles than ideal oracles. However, we show that this leakage does not impact evaluator's privacy.
  First, we formally prove that this leakage $\leakage$ is bounded to providing an oracle $\csk \gets \oracle{Ext}(\cpk)$ which extracts the secret key from HE public keys.
  Next, we prove that the leakage $\leakage$ is independent of the evaluator's private date $Y$.

  We assume that a PPT adversary $\adv$ exists such that $\adv$ gains more advantage from interacting with a cipher oracle instead of an ideal oracle than $\leakage =  \{\oracle{Ext}\}$. We build a new adversary $\adv'$ that given an ideal oracle $\textsf{IDEAL}(Y, \cdot)$ and an extraction oracle $\oracle{Ext}$ can simulate the view of $\adv$.

  \renewcommand{\labelenumii}{(\arabic{enumi}.\alph{enumii})}%
  \begin{enumerate}
    \item Adversary $\adv'$ initiates a new interaction with $\adv$ and receives $(\cpk, c = \encvar{X}) \gets \adv(\cdot)$.
    \item Adversary $\adv'$ uses the extraction oracle within $\leakage$ to extract the secret $\csk \gets \oracle{Ext}(\cpk)$.
    \item Adversary $\adv'$ decrypts $\adv$'s query $X \gets \fheDec(\csk, c)$.
    \item Adversary $\adv'$ verifies if $X$ passes input restrictions. 
    \begin{enumerate}
      \item If any check fails, $\adv'$  sets $A \gets \bot$.
      \item Otherwise,  $\adv'$  interacts with the ideal world oracle with the decrypted input and sets $A \gets \textsf{IDEAL}(X)$.
    \end{enumerate}
    \item Adversary $\adv'$ finishes the execution $\adv(\csk, A)$.
  \end{enumerate}

  We need to show that the adversary $\adv$ cannot distinguish $(\csk, A)$ produced in our simulation from the output of the cipher-world.  Both the simulation and the cipher-world oracle are directly using the extraction oracle $\oracle{Ext}$ to produce the secret key $\csk$ which ensures secret key indistinguishability.
  We study two cases for $A$: (1) The adversary $\adv$ does not follow the input restriction. In this case, the cipher oracle responds with the failure symbol $\bot$. The adversary $\adv'$ performs the same input verification process as the cipher oracle which leads to setting $A \gets \bot$ when one of step (4.a) checks fail. (2) The adversary $\adv$ follows the input restriction. In this scenario, both the cipher and ideal oracles compute the same output ensuring that $A = f(X, Y)$. 
  This proves that the leakage of cipher-world can be bound to $\leakage=\{\oracle{Ext}\}$.

  Now we need to show that the leakage of the cipher-world is independent of the evaluator's private data. 
  As we bound the leakage to an extraction oracle $\{\oracle{Ext}\}$, this independence is clear since extraction is not impacted by changing the evaluator's private data $Y$. 
  Note that in our protocol, exactly one party, owner, generates HE keys, so assuming that HE key extraction is easy has no impact on the evaluator. 
  Therefore, our cipher-world provides the same evaluator privacy guarantee as the ideal world.
\end{proof}

We showed that our cipher-world provides the same server (evaluator) privacy as the original ideal-world.
Note that the cipher-world is one-sided and does not make any security claim about clients (owners).
To prove that our protocol provides server privacy, we have to show that the real-world view can be simulated given access to a cipher-world oracle.
\begin{gather*}
 \left\{\unboundedoracle_{\sims{C}'}^{\text{Client}}(X, \collec)\right\} \compequiv \left\{\textsf{REAL}_{\adv}^{\text{Client}}\left(X, \collec\right)\right\} \\
 \Updownarrow \\
 \sims{C}'(X, \text{rnd},  \adv, \unboundedoracle(Y, \cdot)) \compequiv  \text{View}_{\adv}\left(X, \collec\right) = (X, \text{rnd}, Q, \serverresp, \clientoutput)\
\end{gather*}

We build a simulator $\sims{C}'$ that given a PPT real-world malicious client $\adv$ and a cipher-world oracle $\unboundedoracle(\collec, \cdot)$ fixed with the server's input, simulates the real-world as follows: 
\begin{enumerate}  
  \item Simulator $\sims{C}'$ initiates a new interaction with $\adv$ and receives $(\cpk, Q = \encvar{X'}) \gets \adv(\cdot)$.
  \item Simulator $\sims{C}'$  interacts with the cipher oracle and learns $(\csk, A') \gets \unboundedoracle((\cpk, Q))$.
  \item \label{step:mal-sim-failure} If the cipher-world detects a malicious query $Q$ which does not respect input checks (i.e., $A' = \bot$), the simulator $\sims{C}'$ chooses a random response $t \randin \Zq$, sets the output accordingly $A' \gets \clientfunc(\csk, \encvar{t})$, and skips to the step \ref{step:mal-sim-enc}.
  \item Otherwise, $\sims{C}'$ computes the response $t$ as follows:
  \[
    t \gets
      \left\{\begin{array}{@{\hskip2pt}l@{\hskip12pt}l}
        r \cdot A' & \clientoutput \in \{0, 1\}  \\
        \varAsList{r_j \cdot A'[j]} & \clientoutput \in \{0, 1\}^k  \\
        A' & \clientoutput \in \mathbb{Z}\\
      \end{array}\right. \text{.}
  \]
  Note that the range of $\clientoutput$ is determined by functionality $f$ and is independent of the parties' input.
  \item \label{step:mal-sim-enc} Simulator $\sims{C}'$ performs $\serverresp' \gets \simcircuitp(\csk, f, Q, t)$.
  \item \label{step:mal-sim-a-fail} Adversary $\adv$ only learns the response $\serverresp$ if the decryption succeeds. Therefore, Simulator $\sims{C}'$  checks if the decryption $\fheDec(\csk, \serverresp')$ succeeds. Otherwise, $\sims{C}'$ sets $A' \gets \bot$.
  \item The simulator returns $(X, \text{rnd}, Q = \encvar{X'}, \serverresp', \clientoutput')$.
\end{enumerate}

Now, we show that the real view $(X, \text{rnd}, Q = \encvar{X'}, \serverresp, \clientoutput)$ is indistinguishable from the simulated view $(X, \text{rnd}, Q , \serverresp', \clientoutput')$.
The simulator $\sims{C}'$ does not use variables $X$ and $\text{rnd}$ as there is no guarantee that malicious clients will use their input or random tapes. As $(X, \text{rnd}, Q)$ are directly taken from the input, we only need to show that $(\serverresp, \clientoutput) \compequiv (\serverresp', \clientoutput')$ conditional on the common variables. We split our analysis into three cases:

\emph{Malicious queries which do not represent a set.} When the client is malicious there is no guarantee that the query represents a set. In \cref{sub:base-malicious}, we designed $\malcheckbothfunc$ functions and proved that they randomize the output of our protocols when the query does not represent a set as long as the HE abstraction holds. Our simulator $\sims{C}'$ relies on the cipher oracle to detect when the query does not represent a valid set in step \ref{step:mal-sim-failure} and assigns a uniformly random value $t$ to be encrypted as the response $\serverresp'$. To ensure that $\serverresp'$ has an indistinguishable noise from the real response $\serverresp$,  instead of directly using the encryption in step \ref{step:mal-sim-enc}, $\sims{C}'$ uses the simulator $\simcircuitp$ from the strong input privacy property (\cref{def:str-input-priv}).
Our simulator follows the same $\clientfunc$ process to compute the output $A$ from the response as the real protocol. 
Since $\serverresp \compequiv \serverresp'$, we will have $A \compequiv A'$ as long as the server response $\serverresp'$ decrypts successfully.

\emph{The malicious query decrypts to the set $X'$:} When the query represents a set, $\malcheckbothfunc$ produces an encrypted zero and does not impact the output of the protocol (i.e., adding $\encvar{0}$ is neutral). Knowing that the query is an encryption of the set $X'$, our protocols ensure that $A = f(X', \collec)$ as long as the decryption of $\serverresp$ succeeds. This guarantee follows from combining (1) our semi-honest correctness from \cref{ap:simulation} and (2) knowing that the query decrypts to the same value as $\querybothfunc(\cpk, X')$. The simulator $\sims{C}'$ follows a similar process to compute $\serverresp$ from $A$ as our semi-honest simulator $\sims{C}$ with the difference of replacing $\fheEnc$ with $\simcircuitp$ (i.e., relying on strong input privacy instead of circuit privacy).

\emph{Failed decryption.} Unlike the semi-honest setting where we know that the client's query is freshly encrypted, the server may receive queries with a high noise level. This may lead to producing server responses that fail the decryption. The use of $\simcircuitp$ in step \ref{step:mal-sim-enc} ensures that the response of $\sims{C}'$ has the same noise level as our real-world response. Therefore, the decryption of $\serverresp'$ fails if and only if the decryption of $\serverresp$ fails. When the decryption fails, the real-world cannot compute the output ($A = \bot$) while the simulator $\sims{C}'$ sets $A' \gets \bot$ in step \ref{step:mal-sim-a-fail}. 

\section{Extra benchmarks}
\label{ap:base-bench}
In this section, we provide extra details on our evaluation and add more benchmarks.
First, we compare our small domain PSI layer to an existing small domain paper. Second, we provide more detail on how we design and evaluate generic solutions with the same privacy and functionality as our document search engine and provide extra performance plots. Third, we compare our document search engine to one of the fastest OT-based PSI protocols which does not satisfy our privacy requirements.

\subsection{Small-domain protocols}
\label{ap:sd-bench}
In this section, we evaluate the performance of our small domain PSI-CA protocol and compare it to existing work. We focus on Ruan et al.~\cite{RuanWMZ19} in this section. We compare with Shimizu et al.~\cite{shimizu2015} in \cref{sec:eval}. We do not consider Bay et al.~\cite{BayEAV21} here, as it focuses on a multi-\emph{party} scenario, which leads to higher costs.

The code for the protocol of Ruan et al.~is not available but the paper provides a detailed cost analysis. We use the same scenario and the same CPU (Intel Core i7-7700) to allow us to directly compare performance without requiring us to rerun their protocol from the scratch. Moreover, we extend and optimize their approach for a many-set scenario.  Ruan et al.~use bit-vectors encrypted with ElGamal~\cite{Elgamal85} or Paillier~\cite{Paillier99} encryption to perform PSI. We extend their approach to only compute the query once and apply it to many sets. Their detailed performance benchmark allows us to compute the cost of performing a many-set query with $N$ sets.
In \cref{fig:sd-perf}, we show the computation cost with a fixed input domain size ($|D| \in \{256, 4096\}$) and varied the number of sets. The cost does not include the key exchange. The performance of our scheme is comparable to Run et al. -- which scheme is performing best depends on the specific scenario.

\begin{figure}[tb]
    \centering
    \includegraphics[scale=0.5]{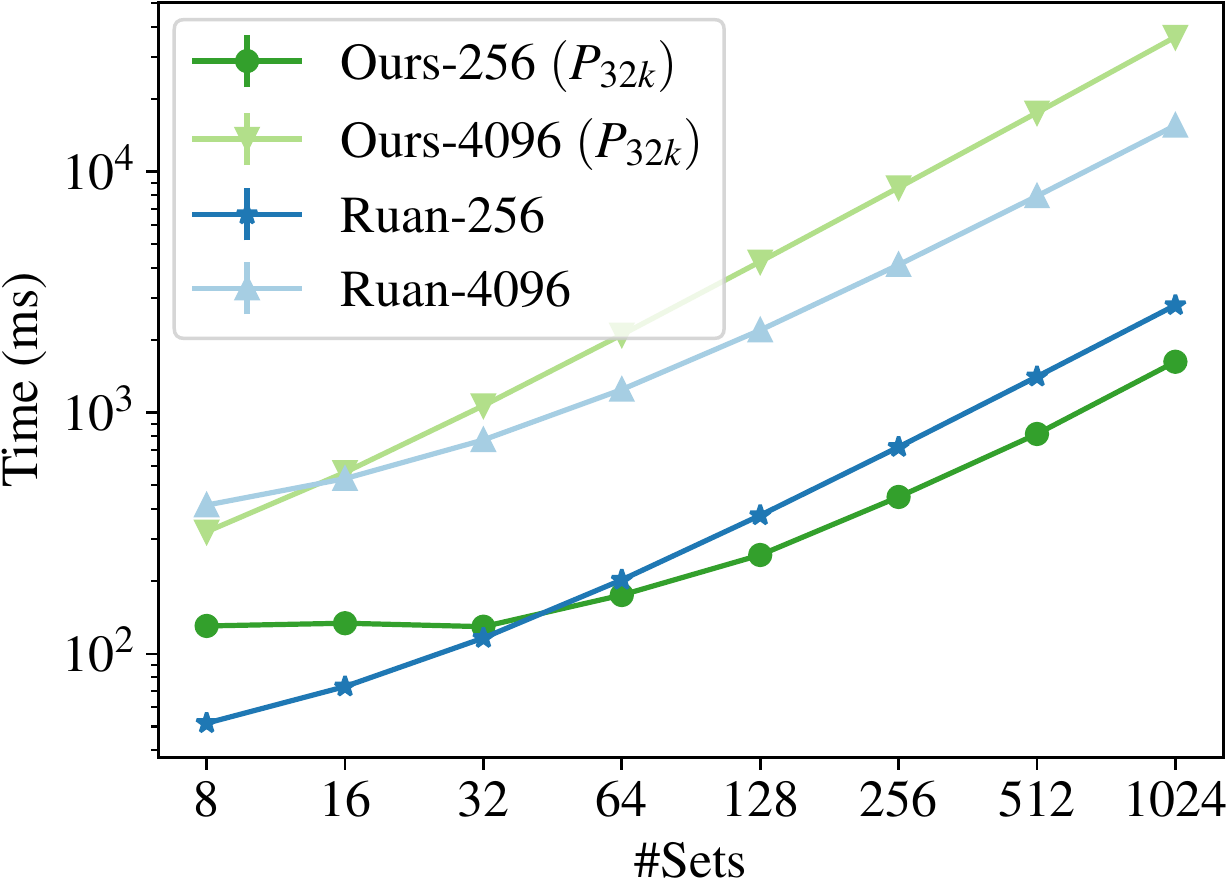}
    \caption{Computation cost for performing small domain PSI-CA Two systems use different threat models: Ruan et al. support 80-bit security in a semi-honest model while ours give 128-bit security against malicious adversaries.}
    \label{fig:sd-perf}
\end{figure}

Despite having the same operating point for both approaches, they have very different security guarantees. Ruan et al. assume a \emph{semi-honest} privacy model and use Pallier keys with 1024-bit RSA primes which only provide \emph{80-bit} security. Extracting more information than cardinality from their protocol is trivial for malicious clients. On the other hand, we provide full \emph{128-bit} security against malicious clients.

\subsection{Circuit-based protocols}
\label{ap:circuit-bench}
In this section, we provide more details on our evaluation of the generic solutions in the document search scenario of \cref{sub:doc-search-eval}.
We expand on their threat model and properties. Moreover, we report and compare the server computation cost of all approaches.

\parasec{Generic SMC} 
We use a high-level SMC compiler, EMP tool-kit~\cite{emp-toolkit}, to design and evaluate circuits providing the same properties as our search engine. More specifically, we use the `EMP-sh2pc' branch of the compiler that provides security in a two-party semi-honest setting and supports both boolean and arithmetic circuits.

We encode each input as a 32 bits binary value, which results in a higher false-positive rate due to encoding keywords than our framework, where we encode keywords using 39 or 44 bits. We use a private equality check offered by EMP for comparing encrypted set elements.
To determine whether one keyword of the client set has a match in a given document, we perform equality tests against all keywords of a document and perform an OR over the comparison results. To perform the full matching, we do an AND over the matching status of all client keywords. This produces a 1-bit result for each document determining its relevance. The F-\match process does not create any extra false-positive in this approach.

Now we have to aggregate $N$ 1-bit document matching statuses according to our X-\agg and CA-\agg policies. The X-\agg variant is straightforward and we use an OR to check if any document is relevant. To count the number of relevant documents, we first convert binary matching statuses to integers encoding `0' or `1' to enable us to continue the computation with an arithmetic circuit. Then, we compute the sum of these integer statuses.

\parasec{Circuit-PSI} 
We choose Chandran et al.~\cite{ChandranGS22} as a state-of-the-art circuit-PSI paper that is secure against semi-honest adversaries. Circuit-PSI protocols perform an intersection between the sets of two parties and secret share the output among them. This enables using circuits to privately compute arbitrary functions over the intersection.
Despite the capability of Chandran et al. to be extended with circuits to compute F-\match matching and X-\agg or CA-\agg aggregation, we decide to not extend their circuit and use the time necessary for computing the intersection as a lower bound on the cost of searching.
Since the PSI protocol of Chandran et al. is a single set protocol, we run $N$ instances of Circuit-PSI sequentially to simulate searching $N$ documents.  
We use the default parameters decided by Chandran et al., meaning that each item is encoded using 32 bits with an additional false positive rate of $2^{-40}$ due to computation.

\parasec{Server's computation cost}
We have already reported the end-to-end latency, communication cost, and client's computation cost in the main body (\cref{fig:doc_search}). We report the server's computation cost in \cref{fig:doc_server_comp}. Starting from 16 documents, our framework has lower server computation than Chandran et al.~\cite{ChandranGS22}. On the other hand, the EMP solutions have better server efficiency than our scheme which is not surprising as our framework outsources the computation load from thin clients to the server. Despite our outsourcing, our CA-\agg search only increases the server cost by a factor of \num{2.5}x and X-\agg by a factor of \num{30}x when searching 1k documents.

\begin{figure}[tb]
    \centering
    \includegraphics[scale=0.5]{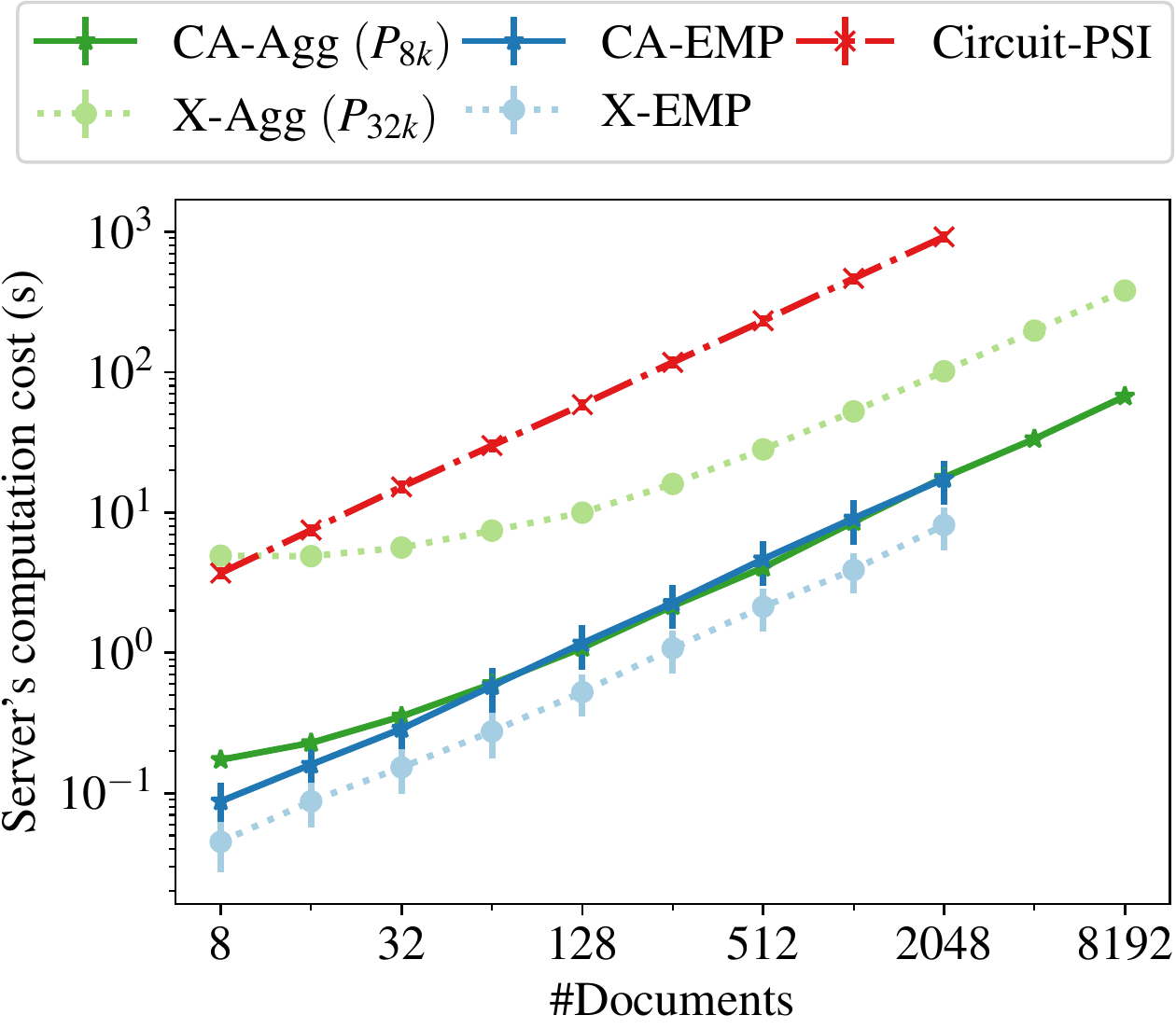}
    \caption{Server's computation cost for document search.}
    \label{fig:doc_server_comp}
\end{figure}

\subsection{OT-based protocol}
\label{ap:spot-bench}
There are efficient PSI protocols that are based on the oblivious transfer in both the semi-honest (such as SpOT-light~\cite{PinkasRTY19}) and the malicious setting (such as PaXoS~\cite{PinkasRTY20}). As discussed in \cref{sec:rel}, this line of research focuses on computing one-to-one equality tests between the client and server which leaks information about each server set and cannot satisfy our privacy requirements. Despite providing a lower privacy guarantee, we compare our approach to the SpOT-light protocol as a baseline cost.

\parasec{Document search with SpOT-light} We follow the document search setting from the \cref{sub:doc-search-eval} and evaluate the cost of using SpOT-light to search $N$ documents. SpOT offers 128-bit security in a semi-honest setting and accepts 256-bit input elements which bypasses the false-positive rate of mapping keywords. However, SpOT is (1) a single-set protocol and (2) does not support privacy extensions such as computing relevance without leaking the intersection cardinality to the client (i.e., private set matching) or aggregating the search result of multiple documents similar to our X-\agg and CA-\agg variant. We handle the single-set limitation by running $N$ sequential PSI interactions to search $N$ documents, but we do not add any countermeasure for the lack of matching or aggregation functionality. 
We encountered concurrency issues (with async IO) when running SpOT. This issue gets amplified when repeating the protocol $N$ times. Instead of directly running the code, we benchmarked the cost of each interaction through multiple runs with $N=8$, then extrapolate the cost for all entries.

\begin{figure}[tb]
    \centering
     \includegraphics[width=0.45\textwidth]{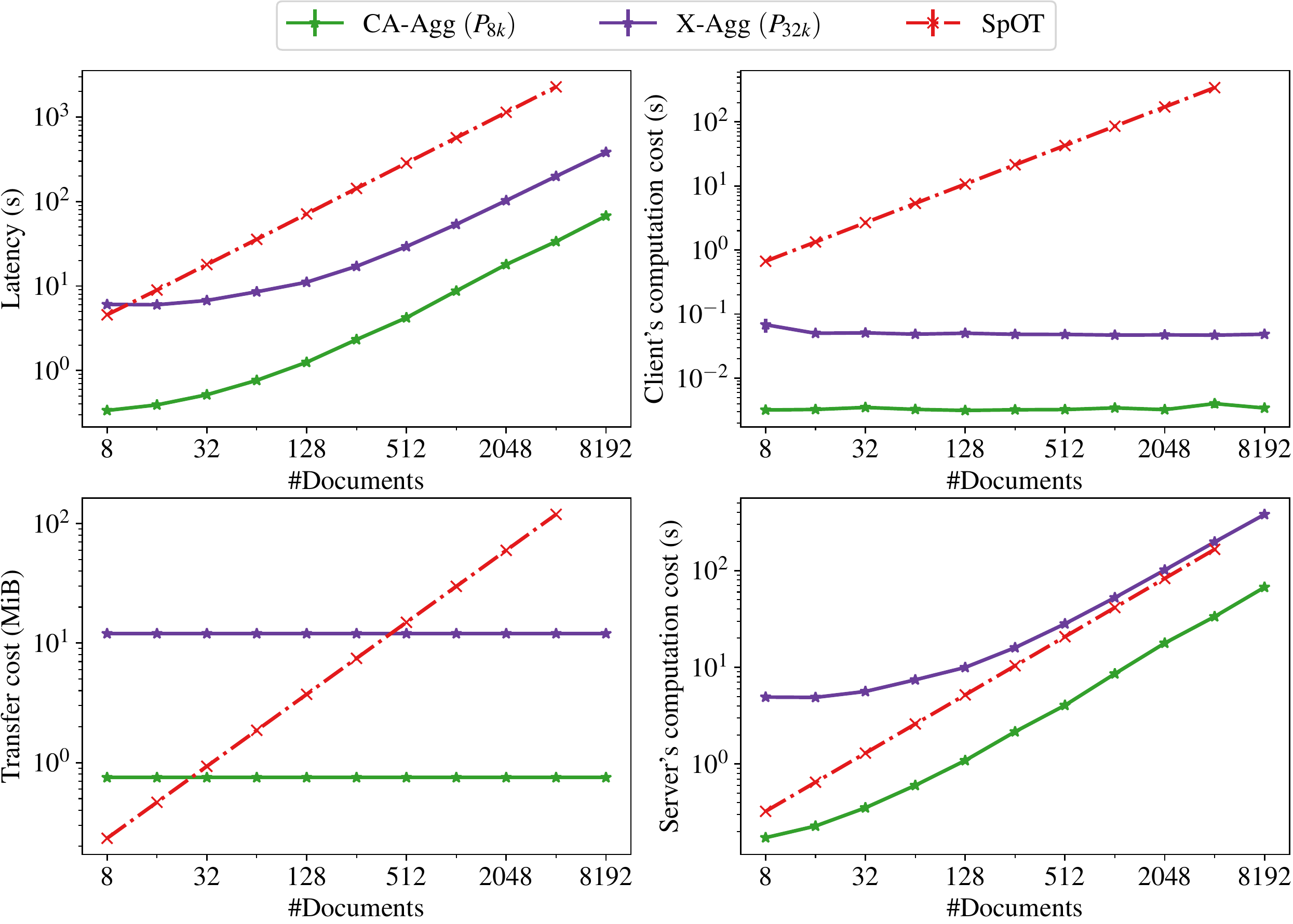}
    \caption{The end-to-end latency (upper-left), communication cost (lower-left), and client and server computation cost (right) of the document search.}
    \label{fig:spot}
\end{figure}

We report the end-to-end latency, computation, and communication costs of SpOT-light in \cref{fig:spot}. 
In the single-set setting and for a small number of documents, SpOT provides better performance.
However, as soon as reaching 32 documents, our CA-\agg variant starts to provide better latency,  computation, and communication than SpOT despite having higher privacy.
When searching 1k documents, our framework improves latency by a factor of \numrange{10}{65}x, communication by a factor of \numrange{1.7}{27}, and client's computation by a factor of \numrange{1800}{24800}x  depending on the search functionality.

\section{Solving Matching in Mobile Apps}
\label{ap:dating}
We do not separately evaluate the matching in the mobile apps scenario as it is similar to the chemical similarity scenario. The set of attributes to be matched can be represented using a small domain as the number of attributes is limited and they have few possible values.
Matching of individual records can be implemented using the threshold matching (Th-\match) protocol. This allows for approximate matches. The results can then be combined using naive aggregation to reveal the matching indices to the querier. Since Th-\match is simpler than the Tv-\match protocol and the threshold for matching is likely smaller than the chemical similarity case, we expect better performance for matching than chemical similarity.

\section{PSI-SUM}
\label{ap:psi-sum}A benefit of our framework's modular design is extensibility. 
To show the ease of adding new functionality, we design a new protocol called PSI-SUM in this section which is getting more popular in the literature due to its use in private ad-monetization systems~\cite{MaC22, IonKNPSSSY17, YingCPXL22}.
In the PSI-SUM protocol, the server assigns a weight to each of its elements and the client wants to compute the sum of weights of common elements, i.e., $\text{PSI-Sum}(X, (Y, W)) = \sum_{\{i | y_i \in X\}} w_i$. 

\begin{algorithm}[tb]
    \footnotesize
    \caption{Adding PSI-SUM capabilities.}
    \label{alg:psi-sum}
    \begin{algorithmic}
        \Function{$\tcserver{\addvariant{PSI-SUM}{\serverfunc}}$}{$\cpk,  Q = \varAsList{\encvar{x_i}}, Y, W = \varAsList{w_i}$}
            \State $\termisinserver{i} \leftarrow  \fheiszero(\cpk, \fheisin(\cpk, y_i, \varAsList{\encvar{x_j}}))$
            \State $\encvar{\mathcal{W}} \leftarrow \sum_{i \in \NatNumUpTo{n}} w_i \cdot \termisinserver{i}$
            \State \Return $\encvar{\mathcal{W}}$
        \EndFunction
        \interalgspace

        \Function{$\tcserver{\addvariant{PSI-SUM-SD}{\serverfunc}}$}{$\cpk,  Q = \varAsList{\encvar{z_i}}, Y, W= \varAsList{w_i}$}
            \State $\varAsList{\termisinserver{i}} \leftarrow \addvariant{PSI-SD}{\serverfunc}(\cpk,  \langle \encvar{{z_i}} \rangle, Y)$
            \State $\encvar{\mathcal{W}} \leftarrow \sum_{d_i \in D}   w_i \cdot \termisinserver{i}$ 
            \State \Return $M \leftarrow \encvar{\mathcal{W}}$
        \EndFunction
        \interalgspace

        \Function{$\tcclient{\addvariant{PSI-SUM}{\clientfunc}}$}{$\cpk,  M = \encvar{\mathcal{W}}$}
            \State \Return  $\fheDec(\csk, \encvar{\mathcal{W}})$
        \EndFunction
    \end{algorithmic}
\end{algorithm}

We add this new protocol in our PSI layer following the structure in \cref{fig:single_psi_proto}. 
We define the PSI-SUM algorithms in \cref{alg:psi-sum}. The $\addvariant{PSI-SUM}{\serverfunc}$ computes a binary inclusion status $\termisinserver{i}$ for each \emph{server} element $y_i$. The server then proceeds similarly to ePSI-CA to compute the encrypted weighted sum. 
The server can continue processing this value in the next layers, such as checking for a threshold value on the sum, or return  $\encvar{W}$ to the client which decrypts it to obtain the answer (see $\addvariant{PSI-SUM}{\clientfunc}$). Our extensions such as malicious check and many-set can directly apply to this protocol without any extra effort.

\end{document}